\documentclass[12pt]{article}

{\begin{Sbox}\begin{minipage}}%
{\end{minipage}\end{Sbox}\fbox{\TheSbox}}%

\usepackage[psamsfonts]{amssymb}
\usepackage{amsmath, amsthm, anysize, enumerate, color, url}
\usepackage{graphicx}
\usepackage{epstopdf}
\usepackage{epsfig}
\usepackage{ulem}

\usepackage{caption,subcaption}
\usepackage{natbib}
\usepackage{threeparttable}

\newtheorem{theorem}{Theorem} 

\newtheorem{lemma}{Lemma} 

\usepackage{breakurl}

\theoremstyle{definition}
\newtheorem{remark}{Remark}  

\newcommand{\E}{\mathbb{E}}
\newcommand{\R}{\mathbb{R}}

\renewcommand{\P}{\mathbb{P}}
\newcommand{\Z}{\mathbb{Z}}

\newcommand{\diag}{\mbox{diag}}
\newcommand{\bs}{\boldsymbol}

\renewcommand{\baselinestretch}{1}
\usepackage{setspace}

\usepackage[top=0.8in, bottom=0.8in, left=0.8in, right=0.8in]{geometry}

\begin{document}

\title{Statistical Inference in a Directed Network Model with Covariates
\footnote{Shortly after finishing the first draft of this paper, we were saddened to hear Steve Fienberg's death. We dedicate this work to his memory.}
}

\author{
Ting Yan\thanks{Department of Statistics, Central China Normal University, Wuhan, 430079, China.
\texttt{Email:} tingyanty@mail.ccnu.edu.cn.}
\hspace{8mm}
Binyan Jiang\thanks{Department of Applied Mathematics, Hong Kong Polytechnic University, Hong Kong. \texttt{Email:}by.jiang@polyu.edu.hk.}
\hspace{8mm}
Stephen E. Fienberg\thanks{Department of Statistics,
Heinz College,
Machine Learning Department,
Cylab,
Carnegie Mellon University,
Pittsburgh, PA 15213, USA. \texttt{Email:} fienberg@stat.cmu.edu.}
\hspace{8mm}
Chenlei Leng\thanks{Corresponding author. Department of Statistics,
University of Warwick and Alan Turing Institute,
Coventry, CV4 7AL,
UK.
\texttt{Email:} C.Leng@warwick.ac.uk.}
}

\date{}

\maketitle

\begin{abstract}
Networks are often characterized by node heterogeneity for which nodes exhibit different degrees of interaction and link homophily for which nodes sharing common features tend to associate with each other. In this paper, we rigorously study a directed network model that captures the former via node-specific parametrization and the latter by incorporating covariates. In particular, this model quantifies the extent of heterogeneity in terms of outgoingness and incomingness of each node by different parameters, thus allowing the number of heterogeneity parameters to be twice the number of nodes.
We study the maximum likelihood estimation of the model and establish the uniform consistency and asymptotic normality of the resulting estimators.
Numerical studies demonstrate our theoretical findings and two data analyses confirm the usefulness of our model.

\end{abstract}



\renewcommand{\baselinestretch}{1.1}\selectfont

\section{Introduction}

Most complex systems involve multiple entities that interact with each other. These interactions are often conveniently represented as networks in which nodes act as entities and a link between two nodes indicates an interaction of some form between the two corresponding entities.
The study of networks has attracted increasing attention in a wide variety of fields including social networks \citep{Burt:Kilduff:Tasselli:2013, Lewisa:Gonzaleza:Kaufmanb:2012}, communication networks \citep{Adamic:Glance:2005, Diesner:Carley:2005}, biological networks \citep{Bader:Hogue:2003, Nepusz:Yu:Paccanaro:2012}, disease transmission networks \citep{Newman:2002} and so on.
Many statistical models have been developed for analyzing networks in the hope to understand their generative mechanism.
However, it remains a unique challenge to understand the statistical properties of many network models; for surveys, see  \cite{Goldenberg2009}, \cite{Fienberg:2012}, and a book long treatment of networks in \cite{Kolaczyk:2009}.

Many networks are characterized by two distinctive features. The first is the so-called \textit{degree heterogeneity} for which nodes exhibit different degrees of interaction.  In the language of \cite{Barabasi:Bonabau:2003}, a typical network often includes a handful of high degree ``hub'' nodes having many edges and many low degree individuals having few edges. The second distinctive feature inherent in most natural and synthetic networks is the so-called \textit{homophily} phenomenon for which links tend to form between nodes sharing common features such as age and sex; see, for example, \cite{McPherson:Lynn:Cook:2001}.  As the name suggests, homophily is best explained by node or link specific covariates used to define similarity between nodes. As a concrete example, we examine the directed friendship network between $71$ lawyers studied in \cite{Lazega:2001} that motivated this paper. The detail of the data can be found in Section 4. As is typical for interactions of this sort, various members' attributes, including formal status (partner or associate), practice (litigation or corporate) etc., are also collected. A major question of interest is whether and how these covariates influence how ties are formed. Towards this end, we plot the network in Figure \ref{figure-data} using red and blue colors to indicate different statuses in (a) and black and green colors to represent lawyers with different practices in (b). To appreciate the difference in the degrees of  connectedness, we use node sizes to represent in-degrees in (a) and out-degrees in (b). This figure highlights a few interesting features. First, there is substantial degree heterogeneity. Different lawyers have different in-degrees and out-degrees, while the in-degrees and the out-degrees of the same lawyers can also be substantially different. This necessitates a model which can characterize the node-specific outgoingness and incomingness. Second, ties seem to form more frequently if the vertices share a common status or a common practice. As a result, a useful model should account for the covariates in order to explain the observed homophily phenomenon.

\begin{figure}[hbt!]
\centering
  \caption{Visualization of Lazega's friendship network among $71$ lawyers.
The vertex sizes
are proportional to either nodal in-degrees in (a) or out-degrees in (b). The positions of the vertices are the same in (a) and (b).
For nodes with degrees less than $5$, we set their sizes the same (as a node with degrees 4). In (a), the colors indicate different statuses (red for partner and blue for associate), while in (b), the colors represent different practices (black for litigation and green for corporate).
}
\begin{subfigure}{.5\textwidth}
  \centering
  \includegraphics[width=\linewidth]{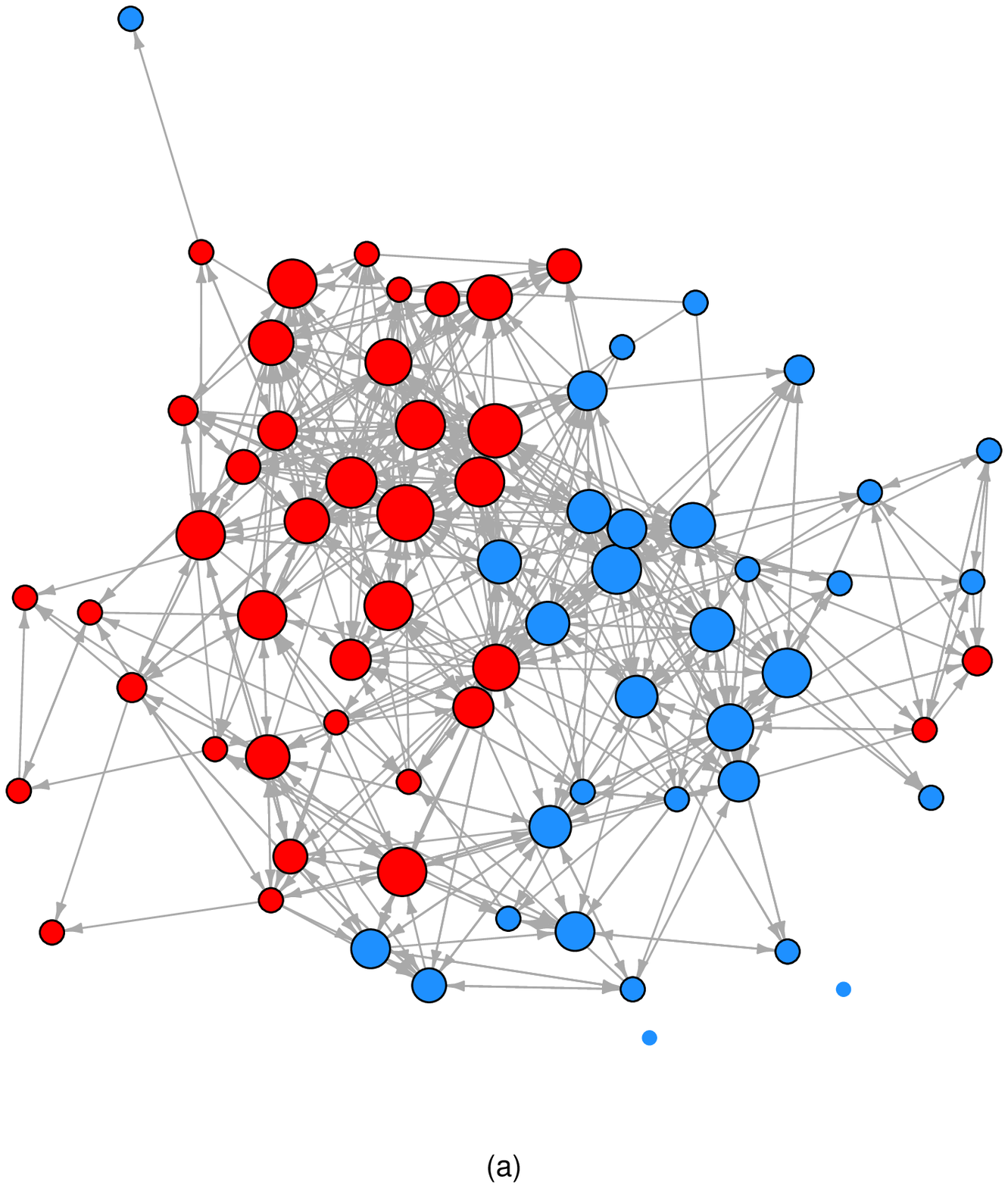}
\end{subfigure}%
\begin{subfigure}{.5\textwidth}
  \centering
  \includegraphics[width=\linewidth]{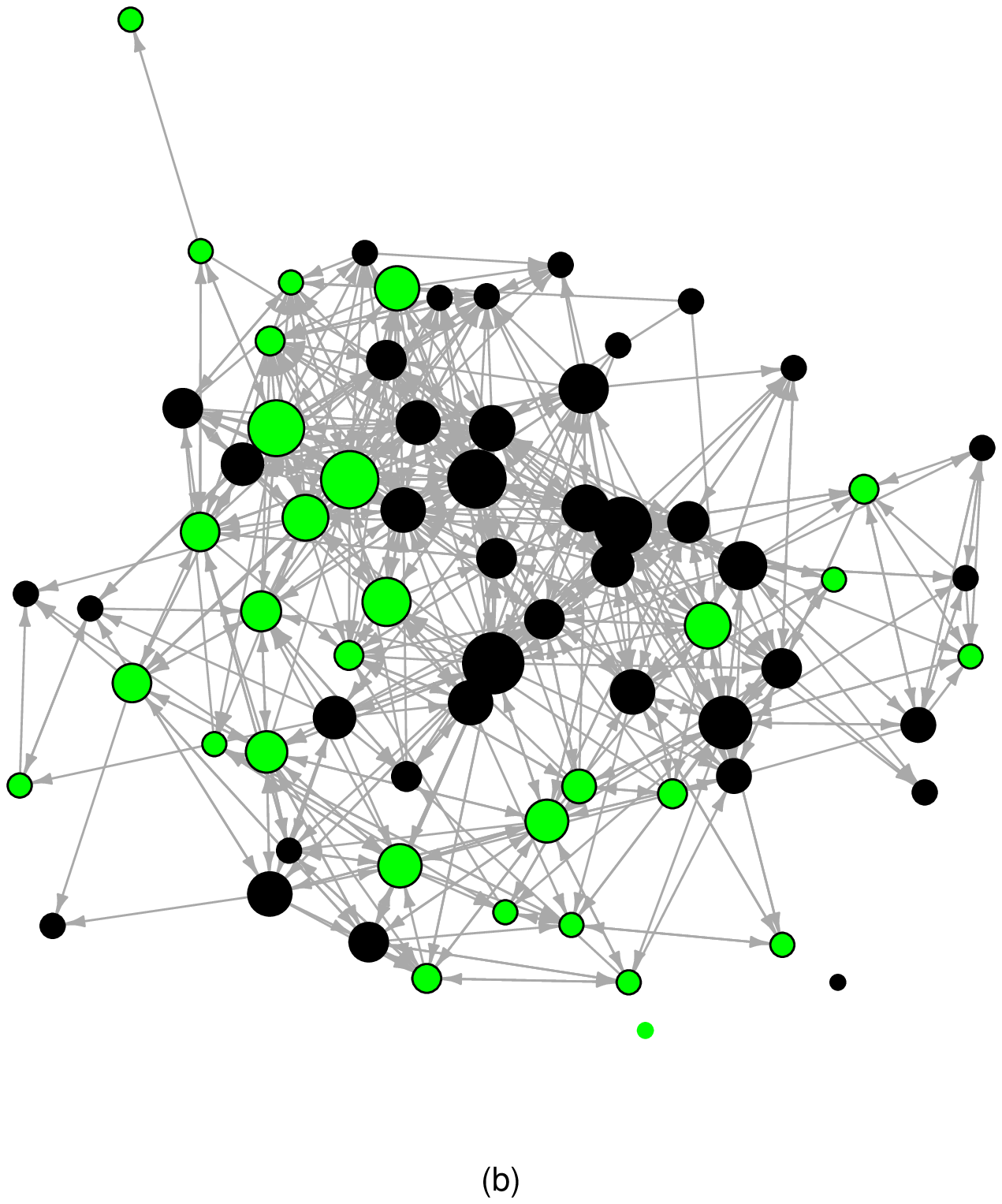}
 \end{subfigure}
\label{figure-data}
\end{figure}

This paper concerns the study of a generative model for directed networks seen in Figure \ref{figure-data} that addresses node heterogeneity and link homophily simultaneously.  Although this model is not entirely new, developing its inference tools is extremely challenging and we have only started to see similar tools for models much simpler when homophily is not considered \citep{Yan:Leng:Zhu:2016}.
Let's start by spelling out the model first. Consider a directed graph $\mathcal{G}_n$ on $n \geq 2$ nodes labeled by $1, \ldots, n$.
Let $a_{ij}\in\{0, 1\}$ be an indictor whether there is a directed edge
from node $i$ pointing to $j$. That is, if there is a directed edge from $i$ to $j$,
then $a_{ij}=1$; otherwise, $a_{ij}=0$.
Denote $A=(a_{ij})_{n\times n}$ as the adjacency matrix of $\mathcal{G}_n$.
We assume that there are no self-loops, i.e., $a_{ii}=0$.
Our model postulates that $a_{ij}$'s follow independent Bernoulli distributions such that a directed link exists from node $i$ to node $j$ with probability
\[
P( a_{ij}= 1) =
\frac{ \exp(  Z_{ij}^\top \bs{\gamma} + \alpha_{i}+ \beta_{j}) }{ 1 + \exp( Z_{ij}^\top \bs{\gamma} + \alpha_{i}+ \beta_{j} ) }.
\]
In this model, the degree heterogeneity of each node is parametrized by two scalar parameters, an incomingness parameter denoted by $\beta_i$ characterizing how attractive the node is and an outgoingness parameter denoted by $\alpha_{i}$ illustrating the extent to which the node is attracted to others \citep{Holland:Leinhardt:1981}. The covariate $Z_{ij}$ is either a link dependent vector or a function of node-specific covariates.  If $X_i$ denotes  a vector of node-level attributes, then these node-level attributes can be used to construct a $p$-dimensional vector $Z_{ij}=g(X_i, X_j)$, where
$g(\cdot, \cdot)$ is a function of its arguments. For instance, if we let $g(X_i, X_j)$ equal to $\|X_i - X_j\|_1$, then it measures the similarity between node $i$ and $j$ features. The vector $\bs{\gamma}$ is an unknown parameter that characterizes the tendency of two nodes to make a connection. Apparently in our model, a larger $Z_{ij}^\top \bs{\gamma}$ implies a higher likelihood for node $i$ and $j$ to be connected. For the friendship network in Figure \ref{figure-data}, for example, the covariate vector may include two covariates, one indicating whether the two nodes share a common status and the other indicating whether their practices are the same.
{
Though similar models for capturing homophily and degree heterogeneity have been considered
by \cite{Dzemski:2014} for a general distribution function and \cite{Graham:2017} in the undirected case,
they focused on the homophily parameter and the inference problem for degree heterogeneity was not studied.
Because the formation of networks is not only influenced by external factors (e.g., dyad covariates), but also affected by
intrinsic factors (e.g., the strengths of nodes to form network connection),  it is statistically interesting to conduct inference on the parameter associated with degree heterogeneity. }

Model \eqref{Eq:density:whole} assumes the independence of the network edges.
As pointed out by \cite{Graham:2017}, the independent assumption may hold in some settings
where the drivers of link formation are predominately bilateral in nature,
as may be true in some trade networks as well as in models of (some types of)
conflict between nation-states.

Since the $n(n-1)$ random variables $a_{i,j}$, $i \neq j$, are mutually independent given the covariates, the probability of observing $\mathcal{G}_n$ is simply
\begin{equation}\label{Eq:density:whole}
 \prod_{i,j=1;i\neq j}^n
 \frac{ \exp \big((Z_{ij}^\top \bs{\gamma} +\alpha_i+\beta_j ) a_{ij} \big)}{1 + \exp( Z_{ij}^\top \bs{\gamma} + \alpha_{i}+ \beta_{j} )}= \exp \big(  \sum_{i,j}a_{ij}Z_{ij}^\top \bs{\gamma} + \bs{\alpha}^\top \mathbf{d} + \bs{\beta}^\top \mathbf{b} - C(\bs{\alpha}, \bs{\beta}, \bs{\gamma} ) \big),
\end{equation}
where
\[C(\bs{\alpha}, \bs{\beta}, \bs{\gamma})=\sum_{i\neq j} \log \big( 1 + \exp( Z_{ij}^\top \bs{\gamma} + \alpha_{i}+ \beta_{j} )\big)\]
is the normalizing constant.
Here $d_{i}= \sum_{j \neq i} a_{ij}$ denotes the out-degree of vertex $i$ and $\mathbf{d}=(d_1, \ldots, d_n)^\top$ is the out-degree sequence of the graph $\mathcal{G}_n$. Similarly, $b_j = \sum_{i \neq j} a_{ij}$ denotes the in-degree of vertex $j$ and $\mathbf{b}=(b_1, \ldots, b_n)^\top$ is the in-degree sequence.  The pair $\{\mathbf{b}, \mathbf{d}\}$ or $\{(b_1, d_1), \ldots, (b_n, d_n)\}$ is the so-called bi-degree sequence. As discussed before, $\bs{\alpha}=(\alpha_1, \ldots, \alpha_n)^\top$ is a parameter vector tied to the out-degree sequence, and $\bs{\beta}=(\beta_1, \ldots, \beta_n)^\top$ is a parameter vector tied to the in-degree sequence,
and $\bs{\gamma}=(\gamma_1, \ldots, \gamma_p)^\top$ is a parameter vector tied to the information of node covariates. Since an out-edge from vertex $i$ pointing to $j$ is the in-edge of $j$ coming from $i$, it is immediate that
the sum of out-degrees is equal to that of in-degrees.
If one transforms $(\bs{\alpha}, \bs{\beta})$ to $(\bs{\alpha}-c, \bs{\beta}+c)$, the likelihood does not change. Because of this,
for the identifiability of the model, we set $\beta_n=0$ as in \cite{Yan:Leng:Zhu:2016}.
{Since we treat $Z_{ij}$ as observed and the likelihood function \eqref{Eq:density:whole} is conditional on $Z_{ij}$'s, we assume that all $Z_{ij}$'s are bounded.
Therefore, the natural parameter space is
\[
\Theta = \{ (\bs{\alpha}^\top, \bs{\beta}_{1, \ldots, n-1}^\top, \bs{\gamma}^\top)^\top:   (\bs{\alpha}^\top, \bs{\beta}_{1, \ldots, n-1}^\top, \bs{\gamma}^\top)^\top \in R^{2n+p-1} \},
\]
under which the normalizing constant is finite.
}

Because of the form of the model and the independent assumption on the links, it appears that maximum likelihood estimation developed for logistic regression is all that is needed for inference. A major challenge of models of this kind is, however, that the number of parameters grows with the network size. In particular, the number of outgoingness and incomingness parameters needed by our model is already twice the size of the network, and the presence of the covariates poses additional challenges. See the literature review below. To a certain extent, our model can be seen as a special case of the exponential random graph model (ERGM) as discussed by \cite{Robins.et.al.2007a, Robins.et.al.2007b}, as the sufficient statistics are the covariates and the bi-degree sequence. It is known, however, that fitting any nontrivial exponential random graph models is extremely challenging, not to mention developing valid procedures for their statistical inference \citep{Goldenberg2009,Fienberg:2012}.
Studying the asymptotic theory of the proposed directed network model is the main contribution of this paper.

We empirically explore the asymptotic properties of the proposed estimators of the heterogeneity parameters $\bs{\alpha}$ and $\bs{\beta}$,
as well as the homophily parameter $\bs{\gamma}$. Our results demonstrate that the empirical study concur with our theoretical findings.
Two real data examples are also provided for illustration.

\subsection{Literature review}
Many network characteristics or configurations can be easily modeled
as exponential family distributions on graphs \citep{Robins.et.al.2007a, Robins.et.al.2007b}. For undirected networks, if we put the node degrees as the sufficient statistics, then the model explains
the observed degree heterogeneity but not homophily. This model is referred to as the $\beta$-model by \cite{Chatterjee:Diaconis:Sly:2011}.
Exploring the properties of the $\beta$-model and its generalizations, however, is nonstandard due to an increasing dimension of the parameter space
and has attracted much recent interest
\citep{Chatterjee:Diaconis:Sly:2011, Perry:Wolfe:2012, Olhede:Wolfe:2012, Hillar:Wibisono:2013, Yan:Xu:2013, Rinaldo2013,  Graham:2017, Karwa:Slakovic:2016}.
In particular, \cite{Chatterjee:Diaconis:Sly:2011} proved the uniform consistency of the maximum likelihood estimator (MLE) and \cite{Yan:Xu:2013} derived
the asymptotic normality of the MLE.
In the directed case, \cite{Yan:Leng:Zhu:2016} studied the MLE of a directed version of the $\beta$-model which is a special case of the $p_1$ model by \cite{Holland:Leinhardt:1981}.
 \cite{Yan:Leng:Zhu:2016} did not consider modelling homophily.
By treating the  node-specific parameters in the $p_1$ model as random effects, \cite{VanDuijnSnijdersZijlstra:2004} proposed a random effects model  incorporating nodal covariates.  The theoretical properties of the MLE of this model are difficult to establish and thus have not been studied.
\cite{Fellows:Handcock:2012} generalized exponential random graph models
by modeling nodal attributes as random
variates. However, the theoretical properties of their model are not explored. \cite{Hoff:2009}  appears to be among the first to study the model in \eqref{Eq:density:whole}. However, the theoretical properties of Hoff's model are again unknown.

It is also worth noting that the consistency and asymptotic normality of the MLE have been derived for two related models:
the Rasch model \citep{Rasch1960} for item response experiments \citep{Haberman1977} and the Bradley-Terry model \citep{bradleyterry52} for paired comparisons by \cite{simons1999} in which a growing number of parameters are modelled.
The data for an item response experiment can be represented as  a bipartite network and for a paired comparisons data as a weighted directed network.
None of these papers discussed how to incorporate covariates.
Finally, Model \eqref{Eq:density:whole} can also be represented as a log-linear model \citep{Fienberg:Rinaldo:2012}.
{Although the necessary and sufficient conditions for
the existence of the MLE for log-linear models with arbitrary dimension have been established [e.g., \cite{Haberman:1974, Fienberg:Rinaldo:2012}],
there is lack of general results on the asymptotic properties of the MLE for high dimensional
log-linear models as the analysis would be challenging [\cite{Erosheva:Fienberg:Joutard:2007, Fienberg:Rinaldo:2007, Fienberg:Rinaldo:2012, Rinaldo:Petrovic:Fienberg:2011}].}

In the above mentioned network models, the dyads of network edges between two nodes are assumed to be mutually independent.
If network configurations such as $k$-stars and triangles are included as sufficient statistics in the ERGMs, then edges are not independent
and such models incur the problem of model degeneracy in the sense of \cite{Handcock:2003}, in which almost all realized
graphs essentially have no edges or are complete, completely skipping all intermediate structures.
\cite{Chatterjee:Diaconis:2013} have shown that most realizations from many ERGMs
 look like the results of a simple Erdos-Renyi model and
given a first rigorous proof of the degeneracy observed in the ERGM with the counts of edges and triangles
as the exclusively sufficient statistics.
\cite{Yin:2015} further gave an explicit characterization of the degenerate tendency as a function of the parameters.
On the other hand, the MLE in ERGMs with dependent structures also incur problematic properties.
\cite{Shalizi:Rinaldo:2013} demonstrated that the MLE is not consistent.
In order to overcome the mode degeneracy in ERGMs, \cite{Schweinberger:Handcock:2015}
have proposed local dependent ERGMs by assuming that the graph nodes can be partitioned into $K$ subsets
(correspondingly, $K$ subgraphs), in which dependence exists within subgraphs and edges are independence
between subgraphs. Based on this assumption, they established a central limit theorem for a network statistic
by referring to the Lindeberg--Feller central limit theorem
when $K$ goes to infinity and the number of nodes in subgraphs is fixed.
The local dependency assumption essentially contains a sequence of independent networks.
On the other hand, some refined network statistics such as ``alternating $k$-stars", ``alternating $k$-triangles" and so on
in \cite{Robins.et.al.2007b} are proposed, but the theoretical properties of the model are still unknown.
Moreover, \cite{Sadeghi:Rinaldo:2014} formalized the ERGM for the joint degree distributions and derived the condition under which
the MLE exists.

The work close to our paper is \cite{Graham:2017} in which the $\beta$-model was generalized to incorporate covariates to explain the homophily phenomenon and degree heterogeneity for undirected networks.  The asymptotic properties of a restricted version of the maximum likelihood
estimator were derived under the assumptions that all parameters are bounded and  that the estimators for
all parameters are taken in one compact set. That is, his results are only applicable to dense networks as pointed out in  \cite{Graham:2017}. In this paper, our focus is on directed networks and our theory is established under more relaxed assumptions. In particular the boundedness assumption on the parameters of degree heterogeneity in \cite{Graham:2017} is not needed in our work. Hence our result covers more general networks. In addition, \cite{Graham:2017} has focused on the consistency and the asymptotic normality of the parameter estimator associated with covariates, while the asymptotic normality of the heterogeneity parameter estimator was not studied. In this paper, we derive these two properties for the covariate parameter and the heterogeneity parameters in model \eqref{Eq:density:whole}. It is worth remarking that establishing the asymptotic normality for estimators of $\bs{\alpha}$ and $\bs{\beta}$ is very challenging with the presence of the covariate $Z$. {\cite{Graham:2016} further proposed a dynamic model to capture
 homophily and transitivity when an undirected network over multiple periods  is observed. The setup is different from ours in that we only observe one network once.}
{
Moreover, \cite{Jochmans:2017} developed a conditional-likelihood based approach to estimate the homophily parameter by constructing a quadruple sufficient statistic to eliminate the degree heterogeneity parameter, and further established the consistency and
asymptotic normality of the resulting estimator.
}

{
To some extent, our network model is connected to  the longitudinal panel data model considered by \cite{FVW2016}  and \cite{Cruz-Gonzalez:Fern¡äandez-Val:Weidner:(2017)} where time and individual fixed effects are both considered. They focused mainly on the homophily parameter.  \cite{Dzemski:2017}  applied the method in \cite{FVW2016}   to a network model similar to ours by including a scalar parameter to characterize the correlation of dyads.  A two-step approach was used for estimation and again the focus is on the homophily parameter. There are major differences between these papers and ours including the methods of proofs, the conditions required by the theorems and the attention to the degree parameters. We will clarify these points after stating our main results in Section 3.
}

For the remainder of the paper, we proceed as follows.
In Section \ref{section:model}, we give the details on the model considered in this paper.
In section \ref{section:main}, we establish asymptotic results.
Numerical studies are presented in Section \ref{section:simulation}.  We provide further discussion and future work in Section \ref{section:summary}.
All proofs are relegated to the appendix.


\section{Maximum Likelihood Estimation}
\label{section:model}

We first introduce some notations.  Let $\R = (-\infty, \infty)$ be the real domain. For a subset $C\subset \R^n$, let $C^0$ and $\overline{C}$ denote the interior and closure of $C$, respectively.  For convenience, let $\bs{\theta}=(\alpha_1, \ldots, \alpha_n, \beta_1, \ldots, \beta_{n-1})^\top$
and $\mathbf{g}=(d_1, \ldots, d_n, b_1, \ldots, b_{n-1})^\top$. Sometimes, we use $\bs{\theta}$ and $(\bs{\alpha}, \bs{\beta})$
interchangeably. For a vector $\mathbf{x}=(x_1, \ldots, x_n)^\top\in \R^n$, denote by
$\|\mathbf{x}\|_\infty = \max_{1\le i\le n} |x_i|$ the $\ell_\infty$-norm of $\mathbf{x}$.  For an $n\times n$ matrix $J=(J_{ij})$, let $\|J\|_\infty$ denote the matrix norm induced by the $\ell_\infty$-norm on vectors in $\R^n$, i.e.
\[
\|J\|_\infty = \max_{\mathbf{x}\neq 0} \frac{ \|J\mathbf{x}\|_\infty }{\|\mathbf{x}\|_\infty}
=\max_{1\le i\le n}\sum_{j=1}^n |J_{ij}|.
\]
The notation $i<j<k$ is a shorthand for $\sum_{i=1}^n \sum_{j=i+1}^n \sum_{k=j+1}^n$.
A ``$*$" superscript on a parameter denotes its true value and may be omitted when doing so causes no
confusion.

In what follows, it is convenient to define the notation:
\[
p_{ij}(\bs{\gamma}, \alpha_i, \beta_j) =
\frac{ \exp( Z_{ij}^\top \bs{\gamma} + \alpha_i + \beta_j) }{ 1 + \exp(Z_{ij}^\top \bs{\gamma} + \alpha_i + \beta_j) }.
\]
The log-likelihood of observing a directed network $\mathcal{G}_n$  under model \eqref{Eq:density:whole} is
\begin{equation}\label{eq:likelihood}
\begin{array}{rcl}
\ell (\bs{\gamma}, \bs{\alpha}, \bs{\beta}) & = & \sum_{i \neq j} \{ a_{ij} \log p_{ij}(\bs{\gamma}, \alpha_i, \beta_j) +
(1- a_{ij}) \log (1-p_{ij}(\bs{\gamma}, \alpha_i, \beta_j)) \} \\
& = & \sum_{i\neq j} a_{ij} Z_{ij}^\top \bs{\gamma} + \sum_{i=1}^n \alpha_i d_i + \sum_{j=1}^n \beta_j b_j - \sum_{i\neq j} \log( 1+ e^{Z_{ij}^\top \bs{\gamma}
+ \alpha_i + \beta_j } ).
\end{array}
\end{equation}
The score equations for the vector parameters $\bs{\gamma}, \bs{\alpha}, \bs{\beta}$ are easily seen as
\renewcommand{\arraystretch}{1.2}
\begin{equation}\label{eq:likelihood-binary}
\large
\begin{array}{rcl}
\sum_{i\neq j} a_{ij}Z_{ij} &=& \sum_{i\neq j} \frac{Z_{ij}e^{Z_{ij}^\top \bs{{\gamma}} + \alpha_i + \beta_j } }
{1+ e^{Z_{ij}^\top \bs{\gamma} + \alpha_i + \beta_j } }, \\
d_i & = & \sum_{k=1, k\neq i}^n \frac{ e^{Z_{ij}^\top \bs{\gamma} + \alpha_i + \beta_k } }
{ 1+ e^{Z_{ij}^\top \bs{\gamma} + \alpha_i + \beta_k}},~~~i=1,\ldots, n, \\
b_j & = & \sum_{k=1,k\neq j}^n \frac{ e^{Z_{ij}^\top \bs{\gamma} + \alpha_k + \beta_j }}
{ 1 + e^{Z_{ij}^\top \bs{\gamma} + \alpha_k + \beta_j } },~~j=1,\ldots, n-1.
\end{array}
\end{equation}
The MLEs of the parameters are the solution of the above equations if they exist.
{
Let $\mathcal{K}$ be the convex hull of the set $\{ ({\bf d}^\top, {\bf b}_{1, \ldots, n-1}^\top, \sum_{i,j} a_{ij}Z_{ij}^\top )^\top: a_{ij}\in \{0, 1\}, 1\le i\neq j \le n \}$.
Since the function $C(\bs{\alpha}, \bs{\beta}, \bs{\gamma})$ is steep and regularly strictly convex, the MLE of $(\bs{\alpha}, \bs{\beta}, \bs{\gamma})$ exists if and only if
$({\bf d}^\top, {\bf b}_{1, \ldots, n-1}^\top, \sum_{i,j} a_{ij}Z_{ij}^\top )^\top$ lies in the interior of $\mathcal{K}$ [see, e.g., Theorem 5.5 in \cite{Brown:1986} (p. 148)].
}
When the number of nodes $n$ is small, we can simply use the R function ``glm" to solve \eqref{eq:likelihood-binary}. For relatively large $n$, this is no longer feasible as it is memory demanding to store the design matrix needed for $\alpha$ and $\beta$. In this case,
we recommend the use of a two-step iterative algorithm by alternating between solving the second and third equations in  \eqref{eq:likelihood-binary} via the fixed point method in \cite{Yan:Leng:Zhu:2016} and
solving the first equation in \eqref{eq:likelihood-binary} via some existing algorithm
for generalized linear models.

In this paper, we assume that $p$, the dimension of $Z$, is fixed and that the support of $Z_{ij}$ is $\Z^p$, where $\Z$ is a compact subset of $\R$.
{For example, if $Z_{ij}$'s are indictor variables such as sex, then the assumption holds.}
For the parameters $\bs{\alpha}$ and $\bs{\beta}$, we make no such assumption and allow them to diverge slowly with $n$, the network size. To be precise, as long as $\| \bs{\theta}^*\|_\infty$, the maximum entry of the true heterogeneity parameter, is bounded by a number proportional to $\log n$, our theory holds. See Theorem \ref{Theorem:binary:con} for example.
For technical reasons, it is more convenient to work with the following restricted
maximum likelihood estimators of $\bs{\alpha}, \bs{\beta}$ and $\bs{\gamma}$ defined as
\begin{equation}\label{definition:RMLE}
(\bs{\widehat{\gamma}}, \bs{\widehat{\alpha}}, \bs{\widehat{\beta}})
=\arg \max_{\bs{\gamma}\in \Gamma, \bs{\alpha}\in \R^n, \bs{\beta}\in \R^{n-1} } \ell(\bs{\gamma}, \bs{\alpha}, \bs{\beta}),
\end{equation}
where $\Gamma$ is a compact subset of $\R^p$ and $\bs{\widehat{\gamma}}=(\hat{\gamma}_1, \ldots,
\hat{\gamma}_p)^\top$,
$\bs{\widehat{\alpha}}=(\hat{\alpha}_1, \ldots, \hat{\alpha}_n)^\top$,
$\bs{\widehat{\beta}}=(\hat{\beta}_1, \ldots, \hat{\beta}_{n-1} )^\top$ are the respective restricted MLEs of
$\bs{\gamma}$, $\bs{\alpha}$ and $\bs{\beta}$,
and $\hat{\beta}_n=0$. Write $\bs{\widehat{\theta}}=( \bs{\widehat{\alpha}}, \bs{\widehat{\beta}})^\top$.
{
 Let $\tilde{\mathcal{K}}$ be the convex hull of the set constructed by all graphical bi-degree sequence $( {\bf d}^\top, {\bf b}_{1, \ldots, n-1}^\top)^\top$ and write
 $( \widehat{\bs{\alpha}}(\bs{\gamma}), \widehat{\bs{\beta}}(\bs{\gamma}))=\arg \min_{\bs{\alpha}, \bs{\beta}} \ell( \bs{\alpha}, \bs{\beta}, \bs{\gamma})$.
For every fixed $\gamma\in \Gamma$, by Theorem 5.5 in \cite{Brown:1986} (p. 148), the MLE $( \widehat{\bs{\alpha}}(\bs{\gamma}), \widehat{\bs{\beta}}(\bs{\gamma}))$ exists if and only if
$( {\bf d}^\top, {\bf b}_{1, \ldots, n-1}^\top)^\top$ lies in the interior of $\tilde{\mathcal{K}}$. Since $\Gamma$ is a compact set, the restricted MLE exists if and only if
$( {\bf d}^\top, {\bf b}_{1, \ldots, n-1}^\top)^\top$ lies in the interior of $\tilde{\mathcal{K}}$.
}

If $\bs{\widehat{\gamma}}$ lies in the interior of $\Gamma$, then it is also the global MLE of $\bs{\gamma}$.
Since we assume the dimension of $Z_{ij}$ is fixed and $\bs{\gamma}$ is one common parameter vector,
it seems reasonable to assume that $\|\bs{\gamma}\|$ is bounded by a constant.
If the restricted MLEs of $\bs{\widehat{\alpha}}$ and $\bs{\widehat{\beta}}$ exist, they would satisfy
the second and third equations in \eqref{eq:likelihood-binary}.
If $\bs{\widehat{\gamma}} \in \Gamma^0$, then it satisfies the first equation in \eqref{eq:likelihood-binary}.
Hereafter, we will work with the MLE defined in \eqref{definition:RMLE} and use ``MLE" to denote ``restricted MLE" for shorthand.

\section{Theoretical Properties}
\label{section:main}

\subsection{Characterization of the Fisher information matrix}

The Fisher information matrix is a key quantity in the asymptotic analysis as it measures the amount of information that a random variable carries about an unknown parameter of a distribution that models the random variable.
In order to characterize this matrix for the vector parameter $\bs{\theta}$ in our model \eqref{Eq:density:whole}, we introduce a general class of matrices that encompass the Fisher matrix. Given two positive numbers $m$ and $M$ with $M \ge m >0$, we say the $(2n-1)\times (2n-1)$ matrix $V=(v_{i,j})$ belongs to the class $\mathcal{L}_{n}(m, M)$ if the following holds:
\begin{equation}\label{eq1}
\begin{array}{l}
m\le v_{i,i}-\sum_{j=n+1}^{2n-1} v_{i,j} \le M, ~~ i=1,\ldots, n-1; ~~~ v_{n,n}=\sum_{j=n+1}^{2n-1} v_{n,j}, \\
v_{i,j}=0, ~~ i,j=1,\ldots,n,~ i\neq j, \\
v_{i,j}=0, ~~ i,j=n+1, \ldots, 2n-1,~ i\neq j,\\
m\le v_{i,j}=v_{j,i} \le M, ~~ i=1,\ldots, n,~ j=n+1,\ldots, 2n-1,~ j\neq n+i, \\
v_{i,n+i}=v_{n+i,i}=0,~~ i=1,\ldots,n-1,\\
v_{i,i}= \sum_{k=1}^n v_{k,i}=\sum_{k=1}^n v_{i,k}, ~~ i=n+1, \ldots, 2n-1.
\end{array}
\end{equation}
Clearly, if $V\in \mathcal{L}_{n}(m, M)$, then $V$ is a $(2n-1)\times (2n-1)$ diagonally dominant, symmetric nonnegative
matrix and $V$ has the following structure:
\[
V= \left(\begin{array}{ll} V_{11} & V_{12} \\
V_{12}^\top  & V_{22}
\end{array}\right) ,
\]
where $V_{11} \in \R^{n \times n}$ and $V_{22} \in \R^{(n-1)\times (n-1)}$ are diagonal matrices, $V_{12}$ is a nonnegative matrix whose non-diagonal elements are positive and diagonal elements equal to zero. One can easily show that the Fisher information matrix for the vector parameter $\bs{\theta}$ belongs to $\mathcal{L}_{n}(m, M)$ for any $\bs{\gamma} \in \Gamma$. The exact form of this matrix can be found after Theorem \ref{Theorem:binary:central}  in Section \ref{subsection:ar}. Thus, with some abuse of notation, we use $V$ to denote the Fisher information matrix for the vector
parameter $\bs{\theta}$ in the model \eqref{Eq:density:whole}.

Define $v_{2n,i}=v_{i,2n}:= v_{i,i}-\sum_{j=1;j\neq i}^{2n-1} v_{i,j}$ for $i=1,\ldots, 2n-1$ and $v_{2n,2n}=\sum_{i=1}^{2n-1} v_{2n,i}$. Then $m \le v_{2n,i} \le M$ for $i=1,\ldots, n-1$, $v_{2n,i}=0$ for $i=n, n+1,\ldots, 2n-1$ and $v_{2n,2n}=\sum_{i=1}^n v_{i, 2n}=\sum_{i=1}^n v_{2n, i}$. Because of the special structure of any matrix $V\in \mathcal{L}_{n}(m, M)$,
\cite{Yan:Leng:Zhu:2016} proposed to approximate its inverse $V^{-1}$ by the matrix $S=(s_{i,j})$, which is defined as
\begin{equation}
\label{definition:S}
s_{i,j}=\left\{\begin{array}{ll}\frac{\delta_{i,j}}{v_{i,i}} + \frac{1}{v_{2n,2n}}, & i,j=1,\ldots,n, \\
-\frac{1}{v_{2n,2n}}, & i=1,\ldots, n,~~ j=n+1,\ldots,2n-1, \\
-\frac{1}{v_{2n,2n}}, & i=n+1,\ldots,2n-1,~~ j=1,\ldots,n, \\
\frac{\delta_{i,j}}{v_{i,i}}+\frac{1}{v_{2n,2n}}, & i,j=n+1,\ldots, 2n-1,
\end{array}
\right.
\end{equation}
where $\delta_{i,j}=1$ when $i=j$ and $\delta_{i,j}=0$ when $i\neq j$.
They established an upper bound on the approximation errors, stated in the lemma below.
\begin{lemma} \label{lemma:inverse:appro}
If $V\in \mathcal{L}_n(m, M)$ with $M/m=o(n)$, then for large enough $n$,
$$\| V^{-1}-S \| \le \frac{c_1M^2}{m^3(n-1)^2},$$
where $c_1$ is a constant that does not depend on $M$, $m$ and $n$, and $\|A\|:=\max_{i,j} |a_{i,j}|$ for a general matrix $A=(a_{i,j})$.
\end{lemma}

This lemma provides an accurate approximation of the inverse of the Fisher information matrix of $\bs{\theta}$ that has a close-form expression. As used throughout our theoretical development, this close-form expression greatly facilitates analytical calculations and makes the covariance matrix in the limiting distribution of the MLE be explicit.

\subsection{Asymptotic results}\label{subsection:ar}

We first establish the existence and consistency of $\widehat{\bs{\theta}}$. The main idea of the proof is as follows.
For every fixed $\bs{\gamma}\in \Gamma$, we define a system of functions
\begin{eqnarray}
\nonumber
F_{\gamma,i}(\bs{\theta}) &  =  & d_i -
\sum_{k=1; k \neq i}^n \frac{e^{Z_{ij}^\top\bs{\gamma}+\alpha_i+\beta_k}}{1+e^{Z_{ij}^\top\bs{\gamma}+\alpha_i+\beta_k} },~~~i=1,\ldots, n, \\
\label{eq:Fgamma}
F_{\gamma,n+j}(\bs{\theta}) & = & b_j - \sum_{k=1; k\neq j}^n \frac{e^{Z_{ij}^\top\bs{\gamma}+\alpha_k+\beta_j}}{1+e^{Z_{ij}^\top\bs{\gamma}+\alpha_k+\beta_j}},~~~j=1,\ldots, n, \\
\nonumber
F_\gamma(\bs{\theta}) & = & (F_{\gamma,1}(\bs{\theta}), \ldots, F_{\gamma,2n-1}(\bs{\theta}))^\top,
\end{eqnarray}
which are just the score equations for $\bs{\theta}$ with $\bs{\gamma}$ fixed.
Then we construct a Newton's iterative sequence $\{\bs{\theta}^{(k+1)}\}_{k=0}^\infty$ with  initial value $\bs{\theta}^{(0)}$,
where $\bs{\theta}^{(k+1)} = \bs{\theta}^{(k)} - [ F'(\bs{\theta}^{(k)})]^{-1} F(\bs{\theta}^{(k)})$.
If the iterative converges, then the solution lies in the neighborhood of $\bs{\theta}_0$.
This is done by establishing a geometrically fast convergence rate of the algorithm with the initial value as the true value.
This  technique is also used in \cite{Yan:Leng:Zhu:2016}.
We first present the consistency of the MLE $\widehat{\bs{\theta}}$ for estimating $\bs{\theta}$ in the following theorem,
whose proof is given in the supplementary material.

\begin{theorem}\label{Theorem:binary:con}
Assume that $\bs{\gamma}^*\in \Gamma^0$ and $\bs{\theta}^*\in \R^{2n-1}$ with $\|\bs{\theta}^*\|_\infty \le \tau \log n $, where $0<\tau<1/24$ is a constant,
and that $A \sim \P_{\bs{\gamma}^*, \bs{\theta}^*}$, where $\P_{\bs{\gamma}^*, \bs{\theta}^*}$ denotes
the probability distribution \eqref{Eq:density:whole} on $A$ under the parameters $\bs{\gamma}^*$ and $\bs{\theta}^*$.
Then as $n$ goes to infinity,
with probability approaching one, the MLE $\widehat{\bs{\theta}}$ exists
and satisfies
\[
\|\widehat{\bs{\theta}} - \bs{\theta}^* \|_\infty = O_p\left( \frac{ (\log n)^{1/2}e^{8\|\bs{\theta}^*\|_\infty} }{ n^{1/2} } \right)=o_p(1).
\]
Further, if $\widehat{\bs{\theta}}$ exists, it is unique.
\end{theorem}

In order to prove the consistency of $\bs{\widehat{\gamma}}$, we define a profile likelihood
\begin{equation}\label{eq:concen-likelihood}
\ell^c(\bs{\gamma}, \bs{ \widehat{\theta} }(\bs{\gamma}))
= \sum_{i\neq j} a_{ij} Z_{ij}^\top \bs{\gamma} + \sum_{i=1}^n \alpha_i(\bs{\gamma}) d_i
+ \sum_{j=1}^n \beta_j(\bs{\gamma}) b_j + \sum_{i\neq j} \log( 1+ e^{Z_{ij}^\top \bs{\gamma}
+ \alpha_i(\bs{\gamma}) + \beta_j(\bs{\gamma}) } ),
\end{equation}
where $\bs{\widehat{\theta}} (\bs{\gamma})=\arg \max_{\bs{\theta}} \ell(\bs{\gamma}, \bs{\theta})$.
It is easy to show that
\begin{equation}\label{eq:exp:like}
\E[\ell(\bs{\gamma}, \bs{\alpha}, \bs{\beta} ) ]
= -\sum_{i\neq j} D_{KL}(p_{ij} \| p_{ij}(\bs{\gamma}, \alpha_i, \beta_j) ) - \sum_{i\neq j} S(p_{ij}),
\end{equation}
where \[D_{KL}(p_{ij}\| p_{ij}(\bs{\gamma}, \alpha_i, \beta_j))=\sum_{i,j} p_{ij} \log \frac{p_{ij} }{p_{ij}(\bs{\gamma}, \alpha_i, \beta_j)}
\]
is the Kullback-Leibler divergence of
$p_{ij}(\bs{\gamma}, \alpha_i, \beta_j)$ from $p_{ij}:=p_{ij}(\bs{\gamma}^*, \alpha_{i}^*, \beta_{j}^*)$
and $S(p)= -p \log p - (1 - p) \log (1 - p)$ is the binary entropy function.
Since the Kullback-Leibler distance is nonnegative, the function \eqref{eq:exp:like} attains its maximum value
when $\bs{\gamma}=\bs{\gamma}^*$, $\bs{\alpha}=\bs{\alpha}^*$ and $\bs{\beta}=\bs{\beta}^*$.
On the other hand, since $p_{ij}$ is a monotonic function on its arguments,
 $(\bs{\gamma}^*, \bs{\alpha}^*, \bs{\beta}^* )$ is a unique maximizer of the function
$ \E[\ell(\bs{\gamma}, \bs{\alpha}, \bs{\beta} ) ]$. The main idea of proving the consistency of $\bs{\widehat{\gamma}}$ is to show that
$n^{-2}| \ell(\bs{\gamma}, \bs{\alpha}, \bs{\beta} ) - \E[\ell(\bs{\gamma}, \bs{\alpha}, \bs{\beta} ) ]|$
is small in contrast with the magnitude of $n^{-2}\E[\ell(\bs{\gamma}, \bs{\alpha}, \bs{\beta} ) ]$, then the MLE approximately
attains at the maximum of the function $\E[\ell(\bs{\gamma}, \bs{\alpha}, \bs{\beta} ) ]$.
The consistency of $\bs{\widehat{\gamma}}$ is stated formally below, whose proof is given in Section \ref{subsection:63:theorem2}.

\begin{theorem}\label{theorem:con-beta}
Assume that $\bs{\gamma}^*\in \Gamma^0$ and $\|\bs{\theta}^*\|_\infty \le \tau \log n $, where $0<\tau<1/24$ is a constant,
and that $A \sim \P_{\bs{\gamma}^*, \bs{\theta}^*}$. Then as $n$ goes to infinity, we have
\[
\bs{\widehat{\gamma}} \stackrel{p}{\longrightarrow} \bs{\gamma}^*.
\]
\end{theorem}

Next, we establish asymptotic normality of $\widehat{\bs{\theta}}$, whose proof is given in the supplementary mateiral.  
This is done by approximately representing
$\bs{\widehat{\theta}}$ as a function of $\mathbf{g}=(d_1, \ldots, d_n, b_1, \ldots, b_{n-1})^\top$ with an explicit expression.

\begin{theorem}\label{Theorem:binary:central}
Assume that $\bs{\gamma}^*\in \Gamma^0$ and $A\sim \P_{\bs{\gamma}^*, \bs{\theta}^*}$.  If $\|\bs{\theta}^*\|_\infty\le \tau\log n$, where $\tau \in (0, 1/44)$ is a constant, then for any fixed $k\ge 1$, as $n \to\infty$, the vector consisting of the first $k$ elements of $(\widehat{\bs{\theta}}-\bs{\theta}^*)$ is asymptotically multivariate normal with mean $\mathbf{0}$ and covariance matrix given by the upper left $k \times k$ block of $S$ defined in \eqref{definition:S}.
\end{theorem}

\begin{remark}
By Theorem \ref{Theorem:binary:central}, for any fixed $i$, as $n\rightarrow \infty$, the convergence rate of $\hat{\theta}_i$ is $1/v_{i,i}^{1/2}$,
whose magnitude is between $O(n^{-1/2}e^{\|\bs{\theta}^*\|_\infty})$ and $O(n^{-1/2})$ by inequality (6) in the supplementary material. 
\end{remark}

Now we provide the exact form of $V$, the Fisher information matrix of the vector parameter $\bs{\theta}$.
For $i=1,\ldots, n$,
\[
v_{i,l} = 0,~~l=1,\ldots, n, ~l\neq i;
~~v_{i,i} =   \sum_{k=1;k \neq i}^n \frac{e^{Z_{ij}^\top\bs{\gamma}+\alpha_i+\beta_k}}{(1+e^{Z_{ij}^\top\bs{\gamma}+\alpha_i+\beta_k})^2},
\]
\[
v_{i,n+j} = \frac{e^{Z_{ij}^\top\bs{\gamma}+\alpha_i+\beta_j}}{(1+e^{Z_{ij}^\top\bs{\gamma}+\alpha_i+\beta_j})^2}, ~~j=1,\ldots, n-1, ~j\neq i;
~~v_{i,n+i}=0
\]
and for $j=1,\ldots, n-1$,
\[
v_{n+j,i} = \frac{e^{Z_{ij}^\top\bs{\gamma}+\alpha_l+\beta_j}}{(1+e^{Z_{ij}^\top\bs{\gamma}+\alpha_l+\beta_j})^2}, ~~l=1,\ldots, n, ~l\neq j; ~~v_{n+j,j} =0,
\]
\[
v_{n+j,n+j} =  \sum_{k=1;k \neq j}^n \frac{e^{Z_{ij}^\top\bs{\gamma}+\alpha_k+\beta_j}}{(1+e^{Z_{ij}^\top\bs{\gamma}+\alpha_k+\beta_j})^2}, ~~v_{n+j, i} = 0, ~~i=1,\ldots, n-1.
\]
Let $H$ be the Hessian matrix of the log-likelihood function $\ell(\bs{\gamma}, \bs{\alpha}, \bs{\beta})$ in \eqref{eq:likelihood} which can be represented as
\[
H=\begin{pmatrix} H_{\gamma\gamma}  & H_{\gamma\theta} \\
H_{\gamma\theta}^\top & -V
\end{pmatrix}.
\]
Following \cite{Amemiya:1985} (p. 126), the Hessian matrix of $\ell^c (\bs{\gamma}^*,
\hat{\theta}(\bs{\gamma}^*))$ is $H_{\gamma\gamma} + H_{\gamma\theta} V^{-1} H_{\gamma\theta}^\top$.
To state the form of the limit distribution of $\hat{\bs{\gamma}}$, define
\begin{equation}\label{eq:I0:beta}
I_n(\bs{\gamma}^*) =  - \frac{1}{n(n-1)} \frac{ \partial^2 \ell^c (\bs{\gamma}^*,
\hat{\theta}(\bs{\gamma}^*)) }{ \partial \bs{\gamma} \partial \bs{\gamma}^\top }
=  \frac{1}{n(n-1)}(-H_{\gamma\gamma} - H_{\gamma\theta} V^{-1} H_{\gamma\theta}^\top),
\end{equation}
whose approximate expression is given in \eqref{eq:I:approximation}, and $I_*(\bs{\gamma})$ as the limit of  $I_n(\bs{\gamma}^*)$
as $n$ goes to infinity.

\begin{theorem}\label{theorem:covariate:asym}
Assume that $\bs{\gamma}^*\in \Gamma^0$ and $\bs{\theta}^*\in \R^{2n-1}$ with $\|\bs{\theta}^*\|_\infty \le \tau \log n $, where $0<\tau<1/24$ is a constant,
and that $A \sim \P_{\bs{\gamma}^*, \bs{\theta}^*}$. Then as $n$ goes to infinity, the $p$-dimensional vector
$N^{1/2}(\bs{\hat{\gamma}}-\bs{\gamma}^* )$ is asymptotically multivariate normal distribution
with mean $I_*^{-1} (\bs{\gamma}) B_*$ and covariance matrix $I_*^{-1}(\bs{\gamma})$,
where $N=n(n-1)$ and $B_*$ is the bias term given in \eqref{definition:Bstar}.
\end{theorem}

\begin{remark}
The limiting distribution of $\bs{\widehat{\gamma}}$ is involved with a bias term
\[
\mu_*=\frac{ I_*^{-1}(\bs{\gamma}) B_* }{ \sqrt{n(n-1)}}.
\]
If all parameters $\bs{\gamma}$ and $\bs{\theta}$ are bounded, then
$\mu_*=O( n^{-1/2})$. It follows that $B_*=O(1)$ and $(I_*)_{i,j}=O(1)$ according to their expressions.
Since the MLE $\bs{\widehat{\gamma}}$ is not centered at the true parameter value, the confidence intervals
and the p-values of hypothesis testing constructed from $\bs{\widehat{\gamma}}$ cannot achieve the nominal level without bias-correction under the null: $\bs{\gamma}^*=0$.
This is referred to as the so-called incidental parameter problem in econometric literature [\cite{Neyman:Scott:1948, FVW2016, Dzemski:2017}].
The produced bias is due to the appearance of additional parameters.
Here, we propose to use the analytical bias correction formula:
$\bs{\widehat{\gamma}}_{bc} = \bs{\widehat{\gamma}}- \hat{I}^{-1} \hat{B}/\sqrt{n(n-1)}$,
where $\hat{I}$ and $\hat{B}$ are the estimates of $I_*$ and $B_*$ by replacing
$\bs{\gamma}$ and $\bs{\theta}$ in their expressions with their MLEs $\bs{\widehat{\gamma}}$ and
$\bs{\widehat{\theta}}$, respectively.
\cite{Dzemski:2014} also used this bias correction procedure, but his expression depends on projected values of pair-wise covariates into
the space spanned by degree parameters $\alpha_i$ and $\beta_j$ under a weighted least square problem and is not explicit.
In the simulation in next section, we can see that the correction formula offer dramatically improvements over uncorrected estimates
and exhibit the corrected coverage probabilities, in which those for uncorrected estimates are below the nominal level evidently.
{See also \cite{Hahn:Newey:2004} and \cite{FVW2016} for Jackknife bias correction for nonlinear panel models. But
as discussed in \cite{Dzemski:2014}, this method is difficult to implement for network models.
Moreover,  \cite{Graham:2017} described an iterated bias
correction procedure,} {which may be numerically unstable and is not guaranteed to converge as demonstrated in \cite{Juodis:2013}.}
\end{remark}

\begin{remark}
There are three main differences between the results in \cite{FVW2016} and those in our paper.
First, for proving their asymptotic results, \cite{FVW2016} used a projection method by projecting the pairwise covariates into the space spanned by degree parameters $\alpha_i$ and $\beta_j$ as a weighted least squares problem,
while we use an elementary method by approximating the inverse matrix of the Fisher information of the degree parameters via an analytical expression.
As a result, the asymptotic variances of the estimators in \cite{FVW2016} depend on
projected values not having closed form expressions, while ours are explicit and easier to compute.
We also note that the matrix to approximate the inverse of the incidental parameter Hessian in  \cite{FVW2016} is diagonal while
ours is not. Second, the asymptotic distribution of the MLE of the incident parameters in $\alpha_i$ and $\beta_j$ is not addressed in \cite{FVW2016}.
Note that the properties of the incidental parameter estimators are more challenging than the fixed dimensional parameter $\gamma$ due to their increasing dimensions.
Third, \cite{FVW2016} assumed that all parameters are bounded while we consider an asymptotic setting to
allow the upper bound of the degree parameter to increase as the size of a network grows.

\end{remark}

\section{Numerical Studies}
\label{section:simulation}
In this section, we evaluate the asymptotic results of the MLEs for model \eqref{Eq:density:whole} through simulation studies and a real data example.

\subsection{Simulation studies}

Similar to \cite{Yan:Leng:Zhu:2016}, the parameter values take a linear form. Specifically,
we set $\alpha_{i+1}^* = (n-1-i)L/(n-1)$ for $i=0, \ldots, n-1$ and let $\beta_i^*=\alpha_i^*$, $i=1, \ldots, n-1$ for simplicity. By default, $\beta_n^*=0$.
We considered four different values for $L$ as $L\in \{0 , \log(\log n), (\log n)^{1/2}, \log n \}$. By allowing the true value of $\bs{\alpha}$ and $\bs{\beta}$ to grow with $n$, we intended to assess the asymptotic properties under different asymptotic regimes.
Similar to \cite{Graham:2017} and \cite{Dzemski:2014},
each element of the $p$-dimensional node-specific covariate $X_i$ is independently generated from a $Beta(2,2)$ distribution.
The difference is that their papers considered $p=1$ while in this paper we set $p=2$ by letting $Z_{ij}=(|X_{i1}-X_{j1}|, |X_{i2}-X_{j2}|)^\top$.
For the parameter $\bs{\gamma}^*$, we let it be $(1, 1.5)^\top$. Thus, the homophily effect of the network is determined by a weighted sum of the similarity measures of the two covariates between two nodes.

Note that by Theorems \ref{Theorem:binary:central}, $\hat{\xi}_{i,j} = [\hat{\alpha}_i-\hat{\alpha}_j-(\alpha_i^*-\alpha_j^*)]/(1/\hat{v}_{i,i}+1/\hat{v}_{j,j})^{1/2}$, $\hat{\zeta}_{i,j} = (\hat{\alpha}_i+\hat{\beta}_j-\alpha_i^*-\beta_j^*)/(1/\hat{v}_{i,i}+1/\hat{v}_{n+j,n+j})^{1/2}$, and $\hat{\eta}_{i,j} = [\hat{\beta}_i-\hat{\beta}_j-(\beta_i^*-\beta_j^*)]/(1/\hat{v}_{n+i,n+i}+1/\hat{v}_{n+j,n+j})^{1/2}$
are all asymptotically distributed as standard normal random variables, where $\hat{v}_{i,i}$ is the estimate of $v_{i,i}$
by replacing $(\bs{\gamma}^*, \bs{\theta}^*)$ with $(\bs{\widehat{\gamma}}, \bs{\widehat{\theta}})$.
Therefore, we assess the asymptotic normality of $\hat{\xi}_{i,j}$, $\hat{\zeta}_{i,j}$ and $\hat{\eta}_{i,j}$ using the quantile-quantile (QQ) plot.  Further, we also record the coverage probability of the 95\% confidence interval, the length of the confidence interval, and the frequency that the MLE does not exist.  The results for $\hat{\xi}_{i,j}$, $\hat{\zeta}_{i,j}$ and $\hat{\eta}_{i,j}$ are similar, thus only the results of $\hat{\xi}_{i,j}$ are reported.
The average and median values of $\bs{\widehat{\gamma}}$ are also reported. Finally,
each simulation is repeated $10,000$ times.

We simulated networks with $n=100$ or $n=200$ and found that the QQ-plots for these two network sizes were similar.
Therefore, we only show the QQ-plots for $n=200$ in Figure \ref{figure-qq} to save space. In this figure,
the horizontal and vertical axes are the theoretical and empirical quantiles, respectively,
and the straight lines correspond to the reference line $y=x$.
In Figure \ref{figure-qq},
when $L=0$ and $\log (\log n)$, the empirical quantiles coincide well with the theoretical ones,
while there are notable deviations when $L=(\log n)^{1/2}$.
When $L=\log n$, the MLE did not exist in all repetitions (see Table \ref{Table:alpha},
thus the corresponding QQ plot could not be shown).

\begin{figure}[!htb]
\centering
\includegraphics[ width=6in, angle=0]{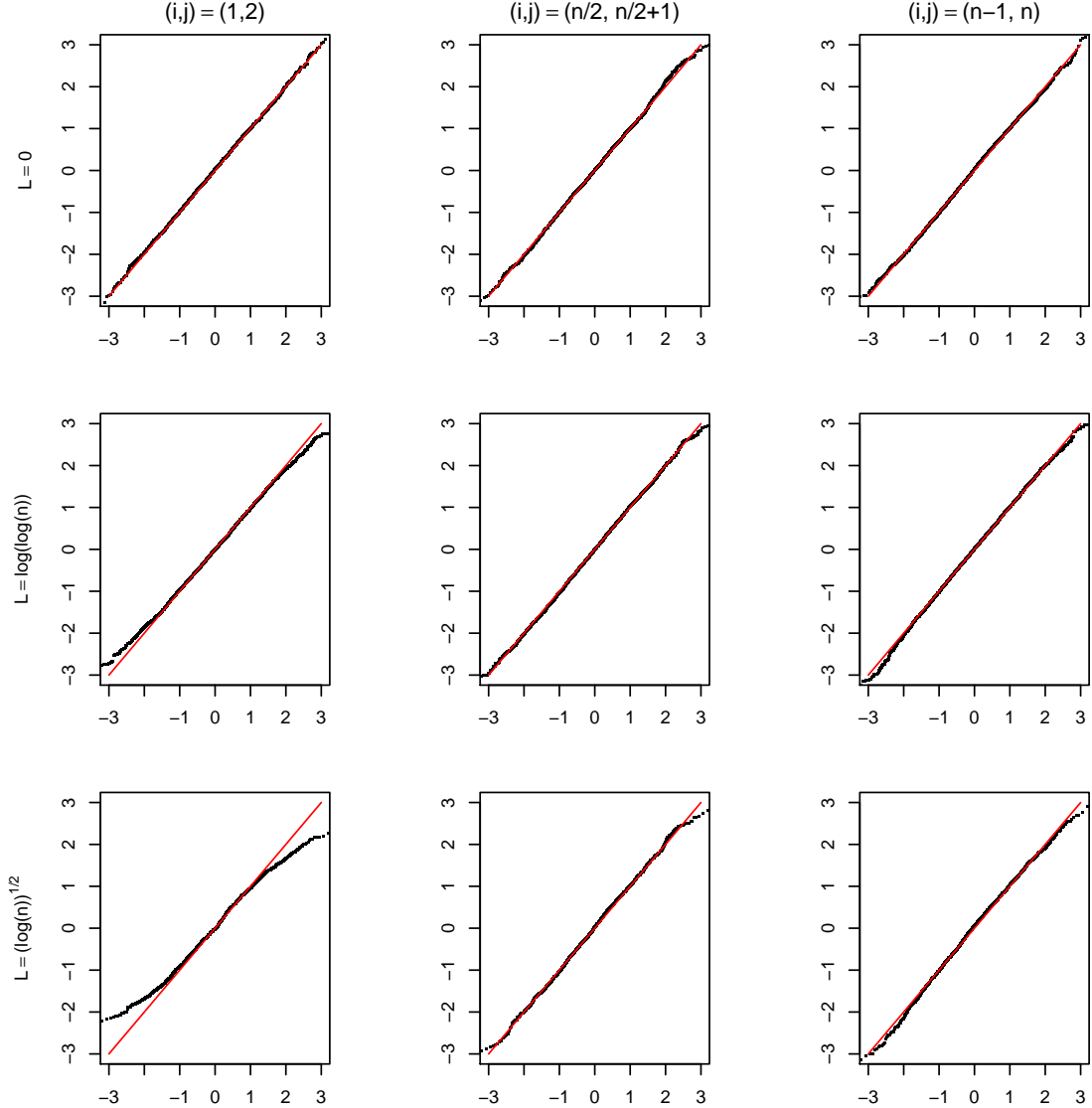}
\caption{The QQ plots of $\hat{v}_{ii}^{1/2}(\hat{\theta}_i-\theta_i)$. }
\label{figure-qq}
\end{figure}

Table \ref{Table:alpha} reports the coverage probability of the 95\% confidence interval for $\alpha_i - \alpha_j$, the length of the confidence interval as well as the frequency that the MLE did not exist.
As we can see, the length of the confidence interval increases as $L$ increases and decreases as $n$ increases, which qualitatively agrees with the theory.  The coverage frequencies are all close to the nominal level when $L=0$ or $\log(\log n)$, while when $L=(\log n)^{1/2}$, the MLE often does not exist and the coverage frequencies for pair $(1, 2)$ are higher than the nominal level; when $L$ is $\log n$, the MLE did not exist for all repetitions.

{\renewcommand{\arraystretch}{1}
\begin{table}[!h]\centering
\caption{The reported values are the coverage frequency ($\times 100\%$) for $\alpha_i-\alpha_j$ for a pair $(i,j)$ / the length of the confidence interval / the frequency ($\times 100\%$) that the MLE did not exist.}
\label{Table:alpha}
\begin{tabular}{ccccccc}
\hline
n       &  $(i,j)$ & $L=0$ & $L=\log(\log n)$ & $L=(\log n)^{1/2}$ & $L=\log n$ \\
\hline
100         &$(1,2)   $&$  94.82/1.20/0$ & $97.02/2.62/0$ & $99.80/3.80/90.04  $ &$ NA/NA/100$ \\
            &$(50,51) $&$  94.76/1.20/0$ & $95.79/1.86/0$ & $96.98/2.37/90.04  $ &$ NA/NA/100$ \\
            &$(99,100)$&$  94.84/1.20/0$ & $95.21/1.44/0$ & $96.38/1.57/90.04  $ &$ NA/NA/100$ \\
&&&&&&\\
200         &$(1,2)     $&$ 95.18/0.84/0 $ & $96.31/1.96/0$ & $98.64/3.05/45.08$ &$ NA/NA/100$ \\
            &$(100,101) $&$ 94.33/0.84/0 $ & $94.88/1.36/0$ & $94.99/1.72/45.08$ &$ NA/NA/100$ \\
            &$(199,200) $&$ 95.08/0.84/0 $ & $94.78/1.02/0$ & $94.95/1.12/45.08$ &$ NA/NA/100$ \\
\hline
\end{tabular}
\end{table}
}

Table \ref{Table:gamma} reports the coverage probabilities for the estimate $\bs{\widehat{\gamma}}$
and bias correction estimate  $\bs{\widehat{\gamma}}_{bc} (= \bs{\widehat{\gamma}}- \hat{I}^{-1} \hat{B}/\sqrt{n(n-1)})$ at the nominal level $95\%$,
the average absolute bias as well as the standard error.
As we can see, the coverage frequencies for the uncorrected estimate is visibly below the nominal level with at least $10$ percentage points
and the bias correction estimate dramatically improve the coverage frequencies, whose coverage frequencies are close to the nominal level when the MLE exists with a high frequency.
On the other hand, when $n$ is fixed, the average absolute bias of $\bs{\widehat{\gamma}}$ increases as $L$ becomes larger and so is the standard error.

{\renewcommand{\arraystretch}{1}
\begin{table}[!htbp]\centering
\caption{
The reported values are the coverage frequency ($\times 100\%$) for $\gamma_i$ for $i$ / average bias / length of confidence interval
 /the frequency ($\times 100\%$) that the MLE did not exist ($\bs{\gamma}^*=(1, 1.5)^\top$).
}
\label{Table:gamma}
\begin{tabular}{cclllcc}
\hline
$n$     &   $\bs{\widehat{\gamma}}$  & $L=0$ & $L=\log(\log n)$ & $L=(\log n)^{1/2}$ & $L=\log n$ \\
\hline
$100$   & $\hat{\gamma}_1$            &$  80.78/0.18/0.56/0  $&$ 5.80/0.81/0.84/0 $&$  0.20/1.28/1.02/90.04$ &  NA \\

        & $\hat{\gamma}_{bc, 1}$      &$  94.28/0.18/0.56/0  $&$ 94.56/0.81/0.84/0 $&$  94.76/1.28/1.02/90.04$ &  NA \\

        & $\hat{\gamma}_2$            &$  81.14/0.19/0.57/0  $&$ 7.31/0.80/0.85/0 $&$  1.41/1.26/1.04/90.04$ &  NA \\

        & $\hat{\gamma}_{bc, 2}$      &$  94.14/0.19/0.57/0  $&$ 94.56/0.80/0.85/0 $&$  93.76/1.26/1.04/90.04$ &  NA\\

$200$   & $\hat{\gamma}_1$            &$  81.23/0.04/0.28/0  $&$ 3.69/0.22/0.43/0  $&$  0.34/0.34/0.52/45.08$  &  NA \\
        & $\hat{\gamma}_{bc, 1}$      &$  95.22/0.04/0.28/0  $&$ 94.37/0.22/0.43/0  $&$  96.19/0.34/0.52/45.08$  &  NA\\

        & $\hat{\gamma}_2$            &$  81.05/0.05/0.28/0  $&$ 4.14/0.22/0.44/0  $&$  0.69/0.34/0.52/45.08$   &  NA \\
        & $\hat{\gamma}_{bc, 2}$      &$  94.38/0.05/0.28/0  $&$ 94.75/0.22/0.44/0  $&$  95.33/0.34/0.53/45.08$   & NA\\

\hline
\end{tabular}
\end{table}
}

\subsection{Two data examples}
{\it The analysis of a Lazega's dataset}.
We first analyze  Lazega's datasets of lawyers \citep{Lazega:2001}, downloaded from \url{https://www.stats.ox.ac.uk/~snijders/siena/Lazega_lawyers_data.htm}.  This data set comes from a network study of corporate law partnership that was carried out in a Northeastern US corporate law firm between 1988 and 1991 in New England. We focus on the friendship network among the $71$ attorneys including partners and associates of this firm. These attorneys were asked to name  attorneys whom they socialized with outside work. Naturally for a network of this sort, many covariates of each attorney were collected. In particular, the collected covariates at the node level include
 formal status (partner or associate); gender (man or woman), location in which they worked (Boston, Hartford, or Providence),
years with the firm, age, practice (litigation or corporate) and law school attended (harvard and yale, or ucon, or others). We define the covariate for each dyad as a $7$ dimensional vector consisting of the differences between these $7$ variables of the two individuals, where for categorical variables, the difference is defined as the indicator whether they are equal, and for continuous variable, the difference indicates their absolute distance.
The directed graph of this data set is shown in Figure \ref{figure-data} where colors indicate either different status in (a) or different practice in (b).
 Although it may deem appropriate to treat the friendship relationship as undirected, from Figure \ref{figure-data},  we can see that the numbers of outgoing and incoming connections for many individuals are dramatically different. As a result, we model the friendship network as a directed one.

{\renewcommand{\arraystretch}{1}
\begin{table}[!hbt]\centering
\caption{The estimators of $\alpha_i$ and $\beta_j$ and their standard errors in the Lazega's data set.}
\label{Table:alphabeta:real}
\begin{tabular}{ccc ccc ccc ccc ccc}
\hline
Vertex & $d_i$  &  $\hat{\alpha}_i$ & $\hat{\sigma}_i$ & $b_j$ &  $\hat{\beta}_i$ & $\hat{\sigma}_j$  && Vertex & $d_i$  &  $\hat{\alpha}_i$ & $\hat{\sigma}_i$ & $b_j$ &  $\hat{\beta}_i$ & $\hat{\sigma}_j$ \\
\hline
1 &$ 4 $&$ -6.21 $&$ 0.63 $&$ 5 $&$ 0.53 $&$ 0.60 $&& 34 &$ 6 $&$ -5.54 $&$ 0.47 $&$ 11 $&$ 1.18 $&$ 0.38 $ \\
2 &$ 4 $&$ -6.01 $&$ 0.67 $&$ 9 $&$ 1.91 $&$ 0.51 $&& 35 &$ 9 $&$ -4.25 $&$ 0.47 $&$ 10 $&$ 1.55 $&$ 0.49 $ \\
4 &$ 14 $&$ -3.46 $&$ 0.44 $&$ 14 $&$ 2.79 $&$ 0.41 $&& 36 &$ 9 $&$ -5.4 $&$ 0.4 $&$ 11 $&$ 0.77 $&$ 0.37 $ \\
5 &$ 3 $&$ -5.01 $&$ 0.64 $&$ 5 $&$ 1.43 $&$ 0.56 $&& 38 &$ 8 $&$ -5.21 $&$ 0.43 $&$ 13 $&$ 1.42 $&$ 0.37 $ \\
7 &$ 1 $&$ -6.59 $&$ 1.06 $&$ 2 $&$ -0.04 $&$ 0.77 $&& 39 &$ 8 $&$ -5.47 $&$ 0.43 $&$ 13 $&$ 1.14 $&$ 0.37 $ \\
8 &$ 1 $&$ -8.32 $&$ 1.06 $&$ 7 $&$ 0.56 $&$ 0.53 $&& 40 &$ 10 $&$ -5.29 $&$ 0.39 $&$ 8 $&$ 0.21 $&$ 0.43 $ \\
9 &$ 6 $&$ -5.98 $&$ 0.55 $&$ 14 $&$ 2.1 $&$ 0.41 $&& 41 &$ 12 $&$ -5.04 $&$ 0.37 $&$ 17 $&$ 1.42 $&$ 0.35 $ \\
10 &$ 14 $&$ -4.17 $&$ 0.44 $&$ 4 $&$ -0.45 $&$ 0.70 $&& 42 &$ 14 $&$ -4.55 $&$ 0.35 $&$ 9 $&$ 0.54 $&$ 0.41 $ \\
11 &$ 5 $&$ -6.49 $&$ 0.56 $&$ 14 $&$ 1.7 $&$ 0.41 $&& 43 &$ 15 $&$ -4.4 $&$ 0.35 $&$ 13 $&$ 1.21 $&$ 0.37 $ \\
12 &$ 22 $&$ -2.95 $&$ 0.38 $&$ 8 $&$ 0.86 $&$ 0.49 $&& 45 &$ 6 $&$ -5.8 $&$ 0.46 $&$ 4 $&$ -0.63 $&$ 0.56 $ \\
13 &$ 14 $&$ -4.35 $&$ 0.42 $&$ 19 $&$ 2.56 $&$ 0.36 $&& 46 &$ 3 $&$ -5.61 $&$ 0.66 $&$ 5 $&$ 0.53 $&$ 0.56 $ \\
14 &$ 6 $&$ -4.27 $&$ 0.51 $&$ 6 $&$ 1.21 $&$ 0.54 $&& 48 &$ 7 $&$ -5.4 $&$ 0.44 $&$ 4 $&$ -0.39 $&$ 0.57 $ \\
15 &$ 3 $&$ -4.89 $&$ 0.64 $&$ 2 $&$ 0.39 $&$ 0.79 $&& 49 &$ 4 $&$ -6.7 $&$ 0.55 $&$ 6 $&$ -0.42 $&$ 0.48 $ \\
16 &$ 8 $&$ -5.66 $&$ 0.48 $&$ 10 $&$ 0.94 $&$ 0.44 $&& 50 &$ 8 $&$ -4.34 $&$ 0.47 $&$ 8 $&$ 1.15 $&$ 0.48 $ \\
17 &$ 23 $&$ -2.85 $&$ 0.37 $&$ 18 $&$ 2.5 $&$ 0.37 $&& 51 &$ 6 $&$ -4.67 $&$ 0.51 $&$ 7 $&$ 1.11 $&$ 0.51 $ \\
18 &$ 8 $&$ -4.62 $&$ 0.46 $&$ 5 $&$ 0.33 $&$ 0.58 $&& 52 &$ 11 $&$ -5.1 $&$ 0.38 $&$ 14 $&$ 1.12 $&$ 0.37 $ \\
19 &$ 4 $&$ -6.85 $&$ 0.59 $&$ 4 $&$ -0.77 $&$ 0.63 $&& 54 &$ 7 $&$ -5.78 $&$ 0.45 $&$ 11 $&$ 0.68 $&$ 0.40 $ \\
20 &$ 12 $&$ -5.01 $&$ 0.43 $&$ 7 $&$ 0.2 $&$ 0.49 $&& 56 &$ 7 $&$ -5.91 $&$ 0.44 $&$ 10 $&$ 0.39 $&$ 0.40 $ \\
21 &$ 8 $&$ -5.73 $&$ 0.46 $&$ 15 $&$ 1.47 $&$ 0.37 $&& 57 &$ 9 $&$ -5.42 $&$ 0.41 $&$ 12 $&$ 0.87 $&$ 0.38 $ \\
22 &$ 8 $&$ -5.67 $&$ 0.44 $&$ 6 $&$ -0.1 $&$ 0.48 $&& 58 &$ 13 $&$ -3.6 $&$ 0.39 $&$ 12 $&$ 1.83 $&$ 0.42 $ \\
23 &$ 1 $&$ -8.65 $&$ 1.05 $&$ 7 $&$ -0.01 $&$ 0.48 $&& 59 &$ 5 $&$ -5.04 $&$ 0.57 $&$ 4 $&$ 0.12 $&$ 0.64 $ \\
24 &$ 23 $&$ -3.59 $&$ 0.34 $&$ 17 $&$ 1.68 $&$ 0.35 $&& 60 &$ 4 $&$ -6.2 $&$ 0.56 $&$ 8 $&$ 0.47 $&$ 0.44 $ \\
25 &$ 11 $&$ -3.95 $&$ 0.41 $&$ 10 $&$ 1.6 $&$ 0.46 $&& 61 &$ 3 $&$ -6.57 $&$ 0.63 $&$ 3 $&$ -0.88 $&$ 0.64 $ \\
26 &$ 9 $&$ -5.45 $&$ 0.43 $&$ 22 $&$ 2.24 $&$ 0.33 $&& 62 &$ 4 $&$ -6.32 $&$ 0.55 $&$ 5 $&$ -0.38 $&$ 0.52 $ \\
27 &$ 13 $&$ -4.54 $&$ 0.38 $&$ 17 $&$ 2.02 $&$ 0.35 $&& 64 &$ 19 $&$ -3.71 $&$ 0.33 $&$ 14 $&$ 1.55 $&$ 0.35 $ \\
28 &$ 11 $&$ -3.91 $&$ 0.42 $&$ 9 $&$ 1.32 $&$ 0.49 $&& 65 &$ 22 $&$ -3.68 $&$ 0.33 $&$ 8 $&$ 0.32 $&$ 0.43 $ \\
29 &$ 10 $&$ -4.81 $&$ 0.39 $&$ 10 $&$ 1.09 $&$ 0.39 $&& 66 &$ 15 $&$ -4.56 $&$ 0.35 $&$ 3 $&$ -0.97 $&$ 0.63 $ \\
30 &$ 6 $&$ -5.26 $&$ 0.53 $&$ 5 $&$ -0.1 $&$ 0.61 $&& 67 &$ 4 $&$ -6.5 $&$ 0.55 $&$ 3 $&$ -1.04 $&$ 0.63 $ \\
31 &$ 25 $&$ -2.21 $&$ 0.33 $&$ 14 $&$ 2.21 $&$ 0.42 $&& 68 &$ 6 $&$ -5.81 $&$ 0.48 $&$ 5 $&$ -0.32 $&$ 0.53 $ \\
32 &$ 4 $&$ -5.86 $&$ 0.63 $&$ 7 $&$ 0.54 $&$ 0.56 $&& 69 &$ 5 $&$ -6.13 $&$ 0.5 $&$ 4 $&$ -0.64 $&$ 0.56 $ \\
33 &$ 12 $&$ -4.03 $&$ 0.42 $&$ 2 $&$ -1.55 $&$ 0.89 $&& 70 &$ 7 $&$ -5.5 $&$ 0.44 $&$ 5 $&$ -0.25 $&$ 0.52 $ \\
34 &$ 6 $&$ -5.54 $&$ 0.47 $&$ 11 $&$ 1.18 $&$ 0.38$\\
\hline
\end{tabular}

\end{table}

In the data set,  individuals are labelled from $1$ to $71$. After removing those individuals whose in-degrees or out-degrees are zeros, we perform the analysis on the $63$ vertices left.
The minimum, $1/4$ quantile, $3/4$ quantile and maximum values of $\mathbf{d}$ are $1$,  $5$,  $8$,  $12$ and $25$, respectively;
those of $\mathbf{b}$ are $2$,  $5$,  $8$, $13$ and $22$, respectively.

The estimators of $\alpha_i$ and $\beta_i$ with their estimated standard errors are given in Table \ref{Table:alphabeta:real}, in which $\beta_{71}=0$ is set as a reference.  The estimates of heterogeneity parameters for in-degrees and out-degrees vary widely:
from the minimum $-7.36$ to maximum $-1.68$ for $\widehat{\alpha}_i$s and from $-1.32$ to $2.56$ for $\widehat{\beta}_i$s.
We then test three null hypotheses $\alpha_1=\alpha_4$, $\alpha_1=\beta_1$ and $\beta_1=\beta_4$, using
the proposed homogeneity test statistics $\hat{\xi}_{i,j} = |\hat{\alpha}_i-\hat{\alpha}_j|/(1/\hat{v}_{i,i}+1/\hat{v}_{j,j})^{1/2}$, $\hat{\zeta}_{i,j} = |\hat{\alpha}_i-\hat{\beta}_j|/(1/\hat{v}_{i,i}+1/\hat{v}_{n+j,n+j})^{1/2}$, and $\hat{\eta}_{i,j} = |\hat{\beta}_i-\hat{\beta}_j|/(1/\hat{v}_{n+i,n+i}+1/\hat{v}_{n+j,n+j})^{1/2}$ respectively. The obtained $p$-values turn out to be $3.5\times 10^{-4}$, $8.7\times 10^{-15}$ and $1.7\times 10^{-3}$, respectively, confirming the need to use our model for parameterizing the  in-degree
and out-degree of each node differently to characterize the heterogeneity of bi-degrees.
The estimated covariate effects, their  bias corrected estimators, their standard errors, and their $p$-values under the null of having no effects are reported in Table \ref{Table:gamma:realdata}. The five categorial variables status, gender, location and practice are all significant and positive, implying that a common value for any of these three variables increases the likelihood of two lawyers to have connection. This is consistent with Figure \ref{figure-data}.
On the other hand, the larger the difference between two lawyers' age or their years with the firm, the less likely they are friends. This makes sense intuitively.
}

{\renewcommand{\arraystretch}{1}
\begin{table}[h]\centering
\caption{The estimators of $\gamma_i$, the corresponding bias corrected estimators, the standard errors, and the $p$-values under the null $\gamma_i=0$ ($i=1, \ldots, 7$) for the Lazega's friendship data.}
\label{Table:gamma:realdata}
\begin{tabular}{ccc ccc cc c}
\hline
Covariate 
&  $\hat{\gamma}_i$ & $\hat{\gamma}_{bc, i}$ & $\hat{\sigma}_i$ &$p$-value  \\
\hline
status          
&  $1.066 $&$ 1.760 $&$ 0.155 $&$ <0.001$\\
gender          
&  $0.580 $&$ 0.962 $&$ 0.142 $&$ <0.001$\\
location        
&  $2.600 $&$ 3.225 $&$ 0.176 $&$ <0.001$\\
years           
&  $-0.108$&$ -0.064$&$ 0.014 $&$ <0.001$\\
age             
&  $-0.040$&$ -0.027$&$ 0.011 $&$ 0.015$ \\
practice        
&  $0.834$ &$ 1.112 $&$ 0.124 $&$ <0.001$ \\
school          
&  $0.267 $&$ -0.479 $&$ 0.123 $&$ <0.001$\\
\hline
\end{tabular}
\end{table}
{\renewcommand{\arraystretch}{1}

{\it The analysis of Sina Weibo data}.
We now analyze the Sina Weibo data collected by \cite{Cai:2018}.
Sina Weibo is the largest Twitter-type social media in China.
The original data contains $4077$ nodes in an official MBA program with directed edges representing who follows who. For our analysis, we first remove those nodes with zero in-degrees or out-degrees since in this case the MLEs of the corresponding degree parameters do not exist. The largest strong connected subgraph of the remaining data set is then examined. This leaves a connected network with $2242$ nodes. The minimum, $1/4$ quantile, $3/4$ quantile and maximum values of $d$ are $1$, $2$, $5$, $19$ and $715$,
respectively; those of $b$ are  $1$,  $4$,  $9$,  $22$,  $253$, respectively. It exhibits a strong degree heterogeneity.

Associated with each node are three variables: the number of characters in personal labels self-created by the users to describe their lifestyles (CHAR),  the cumulated number of Weibo posts (POST), and the time length since Weibo registration measured in months (TIME).
Before our analysis, these node attributes are normalized by subtracting the average and dividing their standard error.
Then the covariates of edges are formed by using the absolute difference distance.

 The two-step
iterative algorithm  in Section \ref{section:model} is used to find the MLEs.
The fitted values of the homophily parameters using model \eqref{Eq:density:whole} are summarized in Table \ref{Table:gamma:realdata2}.
From this table, we can see that all the node attributes are significant. In Figure 1 in the supplementary material, the histograms of the fitted values of the degree parameters are provided. We can see that the estimates of the heterogeneity parameters vary widely: from the minimum of $-2.03$ to the maximum of $4.13$ for $\hat{\beta}_j$'s
and from $-8.87$ to $-1.28$ for $\hat{\alpha}_i$'s. The histogram of $\hat{\beta}_j$'s indicates that $\beta_j$ may follow a normal distribution while that of $\hat{\alpha}_i$'s clearly indicates a skewed distribution.

{\renewcommand{\arraystretch}{1}
\begin{table}[h]\centering
\caption{The estimators of $\gamma_i$, the corresponding bias corrected estimators, the standard errors, and the $p$-values under the null $\gamma_i=0$ ($i=1, 2, 3$) for
the Sina Weibo data.}
\label{Table:gamma:realdata2}
\begin{tabular}{ccc ccc cc c}
\hline
Covariate 
&  $\hat{\gamma}_i$ & $\hat{\gamma}_{bc, i}$ & $\hat{\sigma}_i$ &$p$-value  \\
\hline
CHAR          
&  $0.004 $&$ -0.391 $&$ 0.018 $&$ < 10^{-3}$\\
POST          
&  $0.015 $&$ -0.143 $&$ 0.008  $&$ < 10^{-3}$\\
TIME        
&  $-0.010 $&$ -0.158 $&$ 0.008 $&$ < 10^{-3}$\\
\hline
\end{tabular}
\end{table}
{\renewcommand{\arraystretch}{1}

\section{Discussion}
\label{section:summary}

In this paper, we have derived the consistency and asymptotic normality
of the MLEs for estimating the parameters in model \eqref{Eq:density:whole} when the number of vertices goes to infinity.
{By allowing $\|\bs{\theta}^*\|_\infty$ to diverge to infinity, our model can handle networks with the number of edges growing with the number of node at a slow rate [\cite{Krivitsky:Handcock:Morris:2011}].
If the growth rate on the degree parameters increases too fast, however, the MLE fails to exist with
a positive frequency as demonstrated in the simulation.} Note that
the conditions imposed on $\|\bs{\theta}^*\|_\infty$
in Theorems \ref{Theorem:binary:con}--\ref{theorem:covariate:asym} may not be the best possible.
In particular, the conditions guaranteeing the asymptotic normality seem stronger than those guaranteeing the consistency.  For example, the consistency requires $\|\bs{\theta}^*\|_\infty \le \frac{1}{24}\log n$ while the asymptotic normality requires $\|\bs{\theta}^*\|_\infty \le \frac{1}{44}\log n$.
It would be  interesting to investigate whether these bounds can be improved.

{There is an implicit yet strong assumption for our model that the reciprocity parameter corresponding to the $p_1$-model in \cite{Holland:Leinhardt:1981} is zero.  However, if similarity terms  are included in the model, then there is a tendency toward reciprocity among nodes sharing similar node features. That would alleviate the lack of a reciprocity term to some extent, although it would not induce reciprocity between dissimilar nodes.}
To measure the reciprocity of dyads, it is natural to
incorporate the model term $\rho\sum_{i<j}a_{ij}a_{ji}$ of the $p_1$ model
into \eqref{Eq:density:whole}. In \cite{Yan:Leng:2013}, encouraging empirical results were reported regarding the distribution of the MLE in the $p_1$ model
without covariates. Nevertheless,
although only one new parameter is added, the problem of investigating the asymptotic
theory of the MLEs becomes more challenging. In particular, the Fisher information matrix for the parameter vector $(\rho, \alpha_1,\ldots,\alpha_n, \beta_1, \ldots, \beta_{n-1})$ is not diagonally dominant and thus does not belong to the class $\mathcal{L}_{n}(m, M)$.
In order to apply the method of proofs here, a new approximate matrix with high accuracy of the inverse of the Fisher information matrix is needed.
{On the other hand, various extensions of the $p_1$ model have been developed to allow the reciprocity parameters to depend in a linear fashion on individuals $i$ and $j$ [\cite{Fienberg:Wasserman:1981a}]
and block structures [\cite{Holland:Laskey:Leinhardt:1983, Wang:Wong:1987}]. Though these models may be more realistic, their Fisher information matrices are no longer diagonally dominant. As a result, investigating their asymptotic theory becomes much more involved and we plan to do it in a future work.}

\section*{Acknowledgements}
We are very grateful to three referees, an associate editor, and the Editor for
their valuable comments that have greatly improved the manuscript.
Our simulation code is available on request.
The authors thank Wei Cai at Northeast China Normal University for sharing the Sina Weibo data.
Yan's research is partially supported by the National Natural
Science Foundation of China (No. 11771171). Jiang's research is partially supported by the Hong Kong RGC grant (PolyU 253023/16P).
Leng's research is partially supported by a Turing Fellowship under the EPSRC
grant EP/N510129/1.

\section{Appendix: Proofs for theorems}
\label{section:proofs}
In this section we give the proofs for Theorems \ref{theorem:con-beta} and \ref{theorem:covariate:asym} in
Section \ref{section:main}, and the proofs for Theorems \ref{Theorem:binary:con} and \ref{Theorem:binary:central} are put in the online supplementary material.

\subsection{Proof of Theorem \ref{theorem:con-beta}}
\label{subsection:63:theorem2}

Recall that $\bs{\theta}=(\bs{\alpha}, \bs{\beta})$.
In what follows, the calculations are based on the condition that $\bs{\gamma}\in \Gamma$,
 $\|\bs{\theta}\|_\infty\le n^\tau$, where $\tau\in (0,1/2)$ is a positive constant.
By calculations, we have
\begin{eqnarray*}
\ell(\bs{\gamma}, \bs{\theta} ) & = &   \ell(\bs{\gamma}, \bs{\theta} ) - \E[\ell(\bs{\gamma}, \bs{\theta} )]
+ \E[\ell(\bs{\gamma}, \bs{\theta} )] \\
& = & \sum_{i\neq j} (a_{ij}-p_{ij}) ( Z_{ij}^\top \bs{\gamma} + \alpha_i + \beta_j) + \E[ \ell(\bs{\gamma}, \bs{\theta} )  ],
\end{eqnarray*}
where $\E[ \ell(\bs{\gamma}, \bs{\theta} ) ]$ is given in \eqref{eq:exp:like} and $p_{ij}=p_{ij}(\bs{\gamma}^*, \alpha_i^*, \beta_j^*)$.
By the triangle inequality, we have
\begin{equation}\label{eq:theorem2-tri}
\Big|\frac{1}{n(n-1)} \sum_{i=1}^n \sum_{j\neq i} (a_{ij}-p_{ij})Z_{ij}^\top \bs{\gamma} \Big|
\le \frac{1}{n}\sum_{i=1}^n \Big|\frac{1}{n-1}\sum_{j\neq i} (a_{ij}-p_{ij})Z_{ij}^\top \bs{\gamma} \Big|.
\end{equation}
Since we assume that $Z_{ij}$'s lie in a compact subset of $\R^p$  and {the parameter space $\Theta$ of covariate parameters is compact}, we have for all $i\neq j$,
\begin{equation}\label{eq:kappa}
\max_{\bs{\gamma}\in\Theta} | Z_{ij}^\top \bs{\gamma} | \le \kappa,
\end{equation}
where $\kappa$ is a constant.
By inequality \eqref{eq:kappa}, $a_{ij}Z_{ij}^\top \bs{\gamma}$ is a bounded random variable with the upper bound $\kappa$.
By \citeauthor{Hoeffding:1963}'s (\citeyear{Hoeffding:1963}) inequality, we have
\[
P\Big( \Big|\frac{1}{n-1} \sum_{j\neq i} (a_{ij}-p_{ij}) Z_{ij}^\top \bs{\gamma} \Big| \ge \epsilon \Big ) \le 2 \exp\left( -\frac{ (n-1)\epsilon^2 }{
2\kappa^2  } \right).
\]
By taking $\epsilon=2\kappa [\log (n-1) /(n-1)]^{1/2}$, we have
\[
P\left( \Big|\frac{1}{n-1} \sum_{j\neq i} (a_{ij}-p_{ij}) Z_{ij}^\top \bs{\gamma} \Big| \ge 2\kappa\sqrt{ \frac{\log (n-1)}{(n-1)} } \right ) \le  \frac{4}{(n-1)^2}.
\]
Therefore, we have
\begin{eqnarray*}
 & & P\left( \Big| \frac{1}{n(n-1)} \sum_{i=1}^n \sum_{j\neq i} ( a_{ij}-p_{ij})Z_{ij}^\top \bs{\gamma} \Big|
\ge 2\kappa\sqrt{ \frac{\log (n-1)}{(n-1)} } \right) \\
&\le & P\left( \frac{1}{n}\sum_{i=1}^n \Big|\frac{1}{n-1}\sum_{j\neq i} (a_{ij}-p_{ij})Z_{ij}^\top \bs{\gamma} \Big| \ge 2\kappa\sqrt{ \frac{\log (n-1)}{(n-1)} } \right ) \\
& \le & P\left( \bigcup_{i=1}^n \Big|\frac{1}{n-1}\sum_{j\neq i} (a_{ij}-p_{ij})Z_{ij}^\top \bs{\gamma} \Big| \ge 2\kappa\sqrt{ \frac{\log (n-1)}{(n-1)} } \right) \\
& \le &  \frac{n}{(n-1)^2}.
\end{eqnarray*}
In the above, the first inequality is due to \eqref{eq:theorem2-tri}.
Note that $\|\bs{\alpha}\|\le n^\tau$ and $\|\bs{\beta}\|\le n^\tau$. Similarly,
with probability at most $n/(n-1)^2$, we have
\begin{eqnarray*}
\Big|\frac{1}{n(n-1)} \sum_{i=1}^n \sum_{j\neq i} ( a_{ij}-p_{ij})\alpha_i \Big|
& \ge & \frac{1}{n(n-1)}\sum_{i=1}^n \Big|\sum_{j\neq i} \frac{\alpha_i}{n-1}(a_{ij}-p_{ij}) \Big| \\
&\ge &\frac{1}{n(n-1)}\cdot n \cdot n^\tau \sqrt{ \frac{\log (n-1) }{n-1} }=\frac{ (\log n)^{1/2} }{ n^{1/2-\tau}},
\end{eqnarray*}
and
\[
\Big|\frac{1}{n(n-1)} \sum_{i=1}^n \sum_{j\neq i} ( a_{ij}-p_{ij})\beta_j\Big|
\ge \frac{ (\log n)^{1/2} }{ n^{1/2-\tau}}.
\]
Hence, with probability at least $1-3n/(n-1)^2$, we have
\[
\max_{\bs{\gamma}\le \Gamma, \|\bs{\theta}\|_\infty\le n^\tau } \Big| \frac{1}{n(n-1)} \sum_i \sum_{j\neq i} (a_{ij}-p_{ij}) (Z_{ij}^\top \bs{\gamma} + \alpha_i + \beta_j) \Big|<  \frac{ (\log n)^{1/2} }{ n^{1/2-\tau}},
\]
or equivalently,
\begin{equation}\label{eq:likelihood:error}
\max_{\bs{\gamma}\le \Gamma, \|\bs{\theta}\|_\infty \le n^\tau } \Big| \frac{1}{n(n-1)} \left \{ \ell( \bs{\gamma}, \bs{\theta}) - \E[\ell(\bs{\gamma}, \bs{\theta})] \right \} \Big|
<  \frac{ (\log n)^{1/2} }{ n^{1/2-\tau}}.
\end{equation}

Let $B_n(\rho)=\{ \bs{\gamma}: \|\bs{\gamma} - \bs{\gamma}^*\|_\infty < \rho \}$ be an open ball in $\Gamma$ with $\bs{\gamma}^*$
as its center and $\rho$ as its radius, and $B_n^c(\rho)$ be its complement in $\Gamma$. Define
\[
\epsilon_n(\rho) =\frac{1}{n(n-1)} \left\{ \max_{ \|\bs{\theta}\|_\infty \le n^\tau }  \E[\ell(\bs{\gamma}^*, \bs{\theta} ] -
\max_{\bs{\gamma}\in B_n^c(\rho), \|\bs{\theta}\|_\infty \le n^\tau }
\E[ \ell (\bs{\gamma}, \bs{\theta} )  ]\right \},
\]
and
\[
\epsilon_n (\rho_n) = \arg \min_{\rho} \epsilon_n ( \rho ) >  \frac{ 2(\log n)^{1/2} }{ n^{1/2-\tau} }.
\]
Recall that $\E[\ell(\bs{\gamma}^*, \bs{\theta}) ] = \sum_{i<j} D_{KL}(p_{ij}\| p_{ij}(\bs{\gamma}^*, \alpha_i,  \beta_j) )
- \sum_{i<j} S(p_{ij})$. Therefore,
\begin{eqnarray*}
&&\max_{ \|\bs{\theta}\|_\infty \le n^\tau }  \E[\ell(\bs{\gamma}^*, \bs{\theta} ] -
\max_{\bs{\gamma}\in B_n^c(\rho), \|\bs{\theta}\|_\infty \le n^\tau }
\E[ \ell (\bs{\gamma}, \bs{\theta} )  ]  \\
& = & \max_{ \|\bs{\theta}\|_\infty \le n^\tau }  \sum_{i<j} D_{KL}(p_{ij}\| p_{ij}(\bs{\gamma}^*, \alpha_i,  \beta_j) ) -
\max_{\bs{\gamma}\in B_n^c(\rho), \|\bs{\theta}\|_\infty \le n^\tau }
\sum_{i<j} D_{KL}(p_{ij}\| p_{ij}(\bs{\gamma}^*, \alpha_i,  \beta_j) ).
\end{eqnarray*}
By the property of the Kullback-Leibler divergence
and noticing that $p_{ij}$ is a monotonous function on $\gamma_k$, $\alpha_i$ and $\beta_j$,
$\E[\ell ( \bs{\gamma}, \bs{\theta} )]$ is uniquely maximized at ($\bs{\gamma}^*$, $\bs{\theta}^*$).
Therefore, $\epsilon_n$ will be strictly greater than zero for each fixed $n$.
Further,
since $\epsilon_n(\rho)$ is a continuous increasing function on $\rho$ as $\rho$ increases, we have
\begin{equation}\label{eq:rhon}
\rho_n \to 0, \mbox{~~as~~} n \to\infty.
\end{equation}

Let $E_n$ be the event
\[
\frac{1}{n(n-1)} \Big| \max_{\|\bs{\theta}\|_\infty \le n^\tau }  \ell(\bs{\gamma}, \bs{\theta} )
- \max_{\|\bs{\theta}\|_\infty \le n^\tau }  \E[\ell(\bs{\gamma}, \bs{\theta} )  ] \Big| < \frac{\epsilon_n(\rho_n)}{2}.
\]
for all $\bs{\gamma} \in \Gamma$. Under event $E_n$, we get the inequalities
\begin{equation}\label{inequ:t2:a}
\max_{\|\bs{\theta}_\infty \|\le n^\tau } \frac{1}{n(n-1)} \E[ \ell( \bs{\widehat{\gamma}}, \bs{\theta} ) ]
> \frac{1}{n(n-1)} \ell(\bs{\widehat{\gamma}}, \bs{\widehat{\theta}} ) - \frac{ \epsilon_n(\rho_n) }{2},
\end{equation}
\begin{equation}\label{inequ:t2:b}
\max_{ \|\bs{\theta}\|_\infty \le n^\tau } \frac{1}{n(n-1)} \ell( \bs{\gamma}^*, \bs{\theta} )
> \max_{ \|\bs{\theta}\|_\infty \le n^\tau } \frac{1}{n(n-1)} \E[ \ell(\bs{\gamma}^*, \bs{\theta} )  ] - \frac{ \epsilon_n(\rho_n) }{2}.
\end{equation}
According to the definition of the restricted MLE, we have that
\[
\frac{1}{n(n-1)} \ell( \bs{\widehat{\gamma}}, \bs{\widehat{\theta}} )
\ge \max_{ \|\bs{\theta}\|\le n^\tau } \frac{1}{n(n-1)} \ell( \bs{\widehat{\gamma}}, \bs{\theta} ).
\]
Then, by inequality \eqref{inequ:t2:a}, we have
\begin{equation}\label{inequ:t2:c}
\max_{ \|\bs{\theta}\|_\infty \le n^\tau } \frac{1}{n(n-1)} \E[ \ell (\bs{\widehat{\gamma}}, \bs{\theta} ) ]
> \max_{ \|\bs{\theta}\|_\infty \le n^\tau } \frac{1}{n(n-1)} \ell(\bs{\widehat{\gamma}}, \bs{\theta} ) - \frac{ \epsilon_n}{2}.
\end{equation}
Adding both sides of \eqref{inequ:t2:b} and \eqref{inequ:t2:c} gives
\begin{eqnarray*}
&& \max_{\|\bs{\theta}\|_\infty \le n^\tau } \frac{1}{n(n-1)} \E[ \ell(\bs{\widehat{\gamma}}, \bs{\theta} ) ]
- \Big[ \max_{ \|\bs{\theta}\|_\infty \le n^\tau } \frac{1}{n(n-1)} \ell(\bs{\widehat{\gamma}}, \bs{\theta} ) - \max_{ \|\bs{\theta}\|_\infty \le n^\tau } \frac{1}{n(n-1)} \ell( \bs{\gamma}^*, \bs{\theta} )\Big]
\\
& > & \max_{ \|\bs{\theta}\|_\infty \le n^\tau } \frac{1}{n(n-1)} \E[ \ell(\bs{\gamma}^*, \bs{\theta} ) ] - \epsilon_n(\rho_n)
\\
& = & \max_{\bs{\gamma} \in B_n^c, \|\bs{\theta}\|_\infty \le n^\tau } \frac{1}{n(n-1)} \E[ \ell(\bs{\gamma}, \bs{\theta})  ],
\end{eqnarray*}
where the equality follows the definition of $\epsilon_n$.
By noting that
$$\max_{ \|\bs{\theta}\|_\infty \le n^\tau } \frac{1}{n(n-1)} \ell(\bs{\widehat{\gamma}}, \bs{\theta} ) \ge \max_{ \|\bs{\theta}\|_\infty \le n^\tau } \frac{1}{n(n-1)} \ell( \bs{\gamma}^*, \bs{\theta} ),$$ we have
\[
\max_{\|\bs{\theta}\|_\infty \le n^\tau } \frac{1}{n(n-1)} \E[ \ell(\bs{\widehat{\gamma}}, \bs{\theta} ) ] > \max_{\bs{\gamma} \in B_n^c, \|\bs{\theta}\|_\infty \le n^\tau } \frac{1}{n(n-1)} \E[ \ell(\bs{\gamma}, \bs{\theta})  ].
\]
From the above equation, we have that $E_n \Rightarrow \bs{\widehat{\gamma}}\in B_n(\rho_n)$. Therefore $P(E_n) \le P( \bs{\widehat{\gamma}} \in B_n(\rho_n) )$.
Inequality \eqref{eq:likelihood:error} implies that $\lim_{n\to\infty} P(E_n) =1$
according to the definition of $\rho_n$.
By \eqref{eq:rhon}, it follows that $\bs{\widehat{\gamma}} \stackrel{p}{\to} \bs{\gamma}^*$.

\subsection{Derivation of approximate expression for $I_*(\bs{\gamma})$}
\label{subsection:derivation}
Recall that $H$ is the Hessian matrix of the log-likelihood function \eqref{eq:likelihood}:
\[
H=\begin{pmatrix} H_{\gamma\gamma}  & H_{\gamma\theta} \\
H_{\gamma\theta}^\top & -V
\end{pmatrix},
\]
where
\begin{equation}\label{eq:Hgammagamma}
-H_{\gamma\gamma} = \sum_{i\neq j} p_{ij}(1-p_{ij})Z_{ij}Z_{ij}^\top,
\end{equation}
and
\[
-H_{\gamma \theta}^\top = \begin{pmatrix} \sum_{j\neq 1} p_{1j}(1-p_{1j})Z_{1j}^\top \\
\vdots \\
\sum_{j\neq n}  p_{nj}(1-p_{nj})Z_{nj}^\top \\
\sum_{i\neq 1}  p_{i1}(1-p_{i1})Z_{i1}^\top \\
\vdots \\
\sum_{i\neq n-1}  p_{i,n-1}(1-p_{i,n-1})Z_{i,n-1}^\top
\end{pmatrix}.
\]
In what follows, we will derive the approximate expression
of $I_*(\bs{\gamma})$.
Let $(1)_{m\times n}$ be an $m\times n$ matrix whose elements all are $1$. By calculations, we have
\begin{eqnarray*}
S H_{\gamma\theta}^\top = DH_{\gamma\theta}^\top + \frac{1}{v_{2n,2n}}\begin{pmatrix}
(1)_{n\times n} & (-1)_{n\times (n-1)} \\
(-1)_{(n-1)\times n} & (1)_{(n-1)\times (n-1) }
\end{pmatrix} H_{\gamma\theta}^\top,
\end{eqnarray*}
where $D=\diag(1/v_{11}, \ldots, 1/v_{2n-1, 2n-1})$.
By noting that
\[
\sum_{i=1}^n \sum_{j\neq i} p_{ij}(1-p_{ij})Z_{ij}^\top - \sum_{j=1}^{n-1} \sum_{i\neq j} p_{ij}(1-p_{ij})Z_{ij}^\top
= \sum_{i\neq n} p_{in}(1-p_{in})Z_{in}^\top,
\]
we have
\begin{eqnarray}
\nonumber
H_{\gamma\theta} S H_{\gamma\theta}^\top & = & H_{\gamma\theta} D H_{\gamma\theta}^\top
+ \frac{1}{v_{2n,2n}} H_{\gamma\theta}\begin{pmatrix} (1)_{n\times 1} \\
(-1)_{(n-1)\times 1}
\end{pmatrix}\sum_{i\neq n} p_{in}(1-p_{in})Z_{in}^\top
\\
\nonumber
& = & \sum_{i=1}^n \frac{1}{v_{ii}} \Big(\sum_{j\neq i} p_{ij}(1-p_{ij})Z_{ij}\Big)\Big(\sum_{j\neq i} p_{ij}(1-p_{ij}) Z_{ij}^\top \Big) \\
\label{eq:HSH}
&&+\sum_{j=1}^n \frac{1}{v_{n+j, n+j}} \Big( \sum_{i\neq j} p_{ij}(1-p_{ij})Z_{ij} \Big)\Big(\sum_{i\neq j} p_{ij}(1-p_{ij})Z_{ij}^\top \Big).
\end{eqnarray}
By Lemma \ref{lemma:inverse:appro}, we have
\[
\| V^{-1} - S \| \le \frac{ c_1M^2 }{ m^3(n-1)} \le \frac{c_1}{(n-1)^2}\times \bigg(\frac{1}{4}\bigg)^2 \times
\frac{ (1+e^{2\|\theta^*\|_\infty+\kappa})^6 }{ (e^{2\|\theta^*\|_\infty+\kappa})^3} = O\bigg( \frac{ e^{6\|\theta^*\|_\infty} }{n^2} \bigg).
\]
Therefore,
\[
\|H_{\gamma \theta} ( V^{-1} - S ) H_{\gamma \theta}^\top \|_\infty \le
\|H_{\gamma \theta}\|_\infty^2 \| V^{-1}  - S \|_\infty
\le O(n^2)\times   O\bigg( n \frac{ e^{6\|\theta^*\|_\infty} }{n^2} \bigg) = O( ne^{6\|\theta^*\|_\infty}).
\]
Recall that $N= n(n-1)$ and
note that
\[
(H_{\gamma\gamma} + H_{\gamma \theta} V^{-1} H_{\gamma \theta}^\top ) =
H_{\gamma\gamma} + H_{\gamma \theta} S H_{\gamma \theta}^\top + H_{\gamma \theta} ( V^{-1} - S ) H_{\gamma \theta}^\top.
\]
Therefore, we have
\begin{equation}\label{eq:I:approximation}
-N^{-1} (H_{\gamma\gamma}  +  H_{\gamma \theta} V^{-1} H_{\gamma \theta}^\top)
= -N^{-1} (H_{\gamma\gamma} + H_{\gamma \theta} S H_{\gamma \theta}^\top) + o(1),
\end{equation}
where $H_{\gamma\gamma}$ and $H_{\gamma \theta} S H_{\gamma \theta}^\top$ are given in \eqref{eq:Hgammagamma} and \eqref{eq:HSH}, respectively.
It shows that the limit of $-N^{-1} (H_{\gamma\gamma} + H_{\gamma \theta} S H_{\gamma \theta}^\top)$ is $I_*(\bs{\gamma})$ defined in \eqref{eq:I0:beta}.

\subsection{Proofs for Theorem \ref{theorem:covariate:asym}}
\label{subsection:65:theorem4}

Let $\bs{\hat{\theta}}^*  = \arg\max_{\bs{\theta}} \ell( \bs{\gamma}^*, \bs{\theta})$.
Similar to the proofs of Theorems 1 and 2 in \cite{Yan:Leng:Zhu:2016}, we have two lemmas below,
which will be used in the proof of Theorem \ref{theorem:covariate:asym}.
\begin{lemma}\label{lemma:th4:d}
Assume that $\bs{\theta}^*\in \R^{2n-1}$ with $\|\bs{\theta}^*\|_\infty \le \tau \log n $, where $0<\tau<1/24$ is a constant,
and that $A \sim \P_{\bs{\theta}^*}$. Then as $n$ goes to infinity,
with probability approaching one, the $\bs{\hat{\theta}}^*$ exists
and satisfies
\[
\|\bs{\hat{\theta}}^* - \bs{\theta}^* \|_\infty = O_p\left( \frac{ (\log n)^{1/2}e^{8\|\bs{\theta}^*\|_\infty} }{ n^{1/2} } \right)=o_p(1).
\]
\end{lemma}

\begin{lemma}\label{lemma:th4:c}
If $\|\bs{\theta}^*\|_\infty \le \tau\log n$
and $\tau < 1/40$, then for any $i$,
\[
\hat{\theta}_i^* - \theta_i^* = [S\{ \mathbf{g} - \E(\mathbf{g}) \} ]_i + o_p( n^{-1/2}).
\]
\end{lemma}

For convenience, define $\ell_{ij}(\bs{\gamma}, \bs{\theta})$ by
the $(i,j)^{th}$ dyad's contributions to the log-likelihood function in \eqref{eq:likelihood}, i.e.,
\[
\ell_{ij}( \bs{\gamma}, \bs{\theta})= a_{ij}( Z_{ij}^\top \bs{\gamma} + \alpha_i + \beta_j) - \log( 1 + e^{ Z_{ij}^\top \bs{\gamma} + \alpha_i + \beta_j } ).
\]
Let $T_{ij}$ be a $2n-1$ dimensional vector with ones in its $i$th and $n+j$th elements and zeros otherwise.
Let $s_{\gamma_{ij} }(\bs{\gamma}, \bs{\theta})$ and $s_{\theta_{ij}}(\bs{\gamma}, \bs{\theta})$
denote the score of $\ell_{ij}( \bs{\gamma}, \bs{\theta})$ associated with the
vector parameter $\bs{\gamma}$ and $\bs{\theta}$, respectively:
\[
s_{\gamma_{ij} }(\bs{\gamma}, \bs{\theta}) = \frac{ \partial \ell_{ij} }{ \partial \bs{\gamma} } =
a_{ij} Z_{ij} - \frac{ Z_{ij} e^{ Z_{ij}^\top \bs{\gamma} + \alpha_i + \beta_j } }{ 1 + e^{ Z_{ij}^\top \bs{\gamma} + \alpha_i + \beta_j } },
\]
\[
s_{\theta_{ij}}(\bs{\gamma}, \bs{\theta}) = \frac{ \partial \ell_{ij} }{ \partial \bs{\theta} }
= a_{ij} T_{ij} - \frac{  e^{ Z_{ij}^\top \bs{\gamma} + \alpha_i + \beta_j } }{ 1 + e^{ Z_{ij}^\top \bs{\gamma} + \alpha_i + \beta_j } } T_{ij}.
\]
Then we have the following lemma, whose proof is given in online supplementary material.

\begin{lemma}\label{lemma:th4-b}
Let $H_{\bs{\theta}\bs{\theta}}= -V$ and
\begin{equation}\label{eq:sgamma}
s_{\gamma_{ij}}^* (\bs{\gamma}^*, \bs{\theta}^*) := s_{\gamma_{ij}}( \bs{\gamma}^*, \bs{\theta}^*)
- H_{\bs{\gamma} \bs{\theta} } H_{\bs{\theta}\bs{\theta}}^{-1} s_{\theta_{ij}}(\bs{\gamma}^*, \bs{\theta}^*).
\end{equation}
Then $\frac{1}{\sqrt{N}}[I_n(\bs{\gamma}^* )]^{-1/2} \sum_{i=1}^n \sum_{j\neq i} s_{\gamma_{ij}}^* (\bs{\gamma}^*, \bs{\theta}^*)$
follows asymptotically a $p$-dimensional multivariate standard normal distribution.
\end{lemma}

\begin{proof}[Proof of Theorem \ref{theorem:covariate:asym}]

Recall that $\bs{\widehat{\theta}}(\bs{\gamma})=\arg\max_{\bs{\theta}} \ell( \bs{\gamma}, \bs{\theta})$.
A mean value expansion gives
\[
\sum_{i=1}^n \sum_{j\neq i} s_{\bs{\gamma}_{ij}}(\bs{\widehat{\gamma} }, \bs{\widehat{\theta}} )
- \sum_{i=1}^n \sum_{j\neq i} s_{\bs{\gamma}_{ij}}(\bs{\gamma }^*, \bs{\widehat{\theta}}(\bs{\gamma}^*)  )
=  \sum_{i=1}^n \sum_{j\neq i} \frac{\partial }{ \partial \bs{\gamma}^\top } s_{\bs{\gamma}_{ij}} (\bar{\bs{\gamma}},
\bs{\widehat{\theta}}(\bar{\bs{\gamma}})) (\bs{\widehat{\gamma}}-\bs{\gamma}^*),
\]
where $\bar{\bs{\gamma}}=t\bs{\gamma}^*+(1-t)\bs{\widehat{\gamma}}$ for some $t\in (0, 1)$.
By noting that $\sum_{i=1}^n \sum_{j\neq i} s_{\bs{\gamma}_{ij}}(\bs{\widehat{\gamma} }, \bs{\widehat{\theta}} )=0$, we have
\[
\sqrt{N}(\bs{\widehat{\gamma}}-\bs{\gamma}^*) = -
\Big[ \frac{1}{N} \sum_{i=1}^n \sum_{j\neq i} \frac{\partial }{ \partial \bs{\gamma}^\top } s_{\bs{\gamma}_{ij}} (\bar{\bs{\gamma}},
\bs{\hat{\theta}}(\bar{\bs{\gamma}})) \Big]^{-1}
\times \Big[\frac{1}{\sqrt{N}} \sum_{i=1}^n \sum_{j\neq i} s_{\bs{\gamma}_{ij}}(\bs{\gamma}^*, \bs{\hat{\theta}}(\bs{\gamma}^*) ) \Big].
\]
Since the dimension $p$ of $\bs{\gamma}$ is fixed, by Theorem \ref{theorem:con-beta}, we have
\[
-\frac{1}{N} \sum_{i=1}^n \sum_{j\neq i} \frac{\partial}{ \partial \bs{\gamma}^\top } s_{\bs{\gamma}_{ij}}
(\bar{\bs \gamma}, \bs{\widehat{\theta}}(\bar{\bs \gamma}) )
\stackrel{p}{\to }  I_*(\bs{\gamma}).
\]
Let $\bs{\hat{\theta}}^*=\bs{\widehat{\theta}}(\bs{\gamma}^*)$. Therefore,
\begin{equation}\label{eq:theorem4:aa}
\sqrt{N} (\bs{\widehat{\gamma}}-\bs{\gamma}^*) = I_*^{-1}(\bs{\gamma}) \times \Big[ \frac{1}{\sqrt{N}} \sum_{i=1}^n \sum_{j\neq i}
s_{\gamma_{ij}} ( \bs{\gamma}^*, \bs{\hat{\theta}}^*  )\Big] + o_p(1).
\end{equation}
By applying a third order Taylor expansion to the summation in brackets in \eqref{eq:theorem4:aa}, it yields
\begin{equation}\label{eq:gamma:asym:key}
\frac{1}{\sqrt{N}} \sum_{i=1}^n \sum_{j\neq i} s_{\gamma_{ij}} ( \bs{\gamma}^*,
\bs{\hat{\theta}}^* ) = S_1 + S_2 + S_3,
\end{equation}
where
\begin{equation*}
\begin{array}{l}
S_1  =  \frac{1}{\sqrt{N}} \sum_{i=1}^n \sum_{j\neq i} s_{\gamma_{ij}} ( \bs{\gamma}^*, \bs{\theta}^*)
+ \frac{1}{\sqrt{N}} \sum_{i=1}^n \sum_{j\neq i}
\Big[\frac{\partial}{\partial \bs{\theta}^\top } s_{\gamma_{ij}} (\bs{\gamma}^*, \bs{\theta}^*)\Big]( \bs{\hat{\theta}}^* - \bs{\theta}^* ), \\
S_2  =   \frac{1}{2\sqrt{N}} \sum_{k=1}^{2n-1} \Big[( \hat{\theta}_k^* - \theta_k^* ) \sum_{i=1}^n \sum_{j\neq i}
\frac{\partial^2 }{ \partial \theta_k \partial \bs{\theta}^\top } s_{\gamma_{ij}} ( \bs{\gamma}^*, \bs{\theta}^*)
\times ( \bs{\hat{\theta}}^* - \bs{\theta}^* ) \Big],  \\
S_3  =  \frac{1}{6\sqrt{N}} \sum_{k=1}^{2n-1} \sum_{l=1}^{2n-1} \{ (\hat{\theta}_k^* - \theta_k^*)(\hat{\theta}_l^* - \theta_l^*)
\Big[ \sum_{i=1}^n \sum_{j\neq i} \frac{ \partial^3 s_{\gamma_{ij}} (\bs{\gamma}^*, \bar{\theta}^* )}{ \partial \theta_k \partial \theta_l \partial \bs{\theta}^\top } \Big]
(\bs{\hat{\theta}}^*  - \bs{\theta}^* )\}.
\end{array}
\end{equation*}
Similar to the proof of Theorem 4 in \cite{Graham:2017}, we will show that (1) $S_1$ is asymptotically normal distribution;
(2) $S_2$ is the bias term having a non-zero probability limit; (3)$S_3$ is an asymptotically negligible
remainder term.

We work with $S_1$, $S_2$ and $S_3$ in reverse order.
We first evaluate the term $S_3$.
We calculate $g^{ij}_{klh}=\frac{ \partial^3 s_{\gamma_{ij}} (\gamma, \theta ) }{ \partial \theta_k \partial \theta_l \partial \theta_h }$ as follows.\\
(1) For different $k, l, h$, $g^{ij}_{klh}=0$. \\
(2) Only two values are equal. If $k=l=i\le n; h\ge n+1$,
$g^{ij}_{klh}=p_{ij}(1-p_{ij})(1-6p_{ij}+6p_{ij}^2)Z_{ij}$; for other cases, the results are similar.\\
(3) Three values are equal.
$g^{ij}_{klh}=p_{ij}(1-p_{ij})(1-6p_{ij}+6p_{ij}^2)Z_{ij}$ if $k=l=h=i\le n$;
$g^{ij}_{klh}=p_{ji}(1-p_{ji})(1-6p_{ji}+6p_{ji}^2)Z_{ji}$ if $k=l=h=j \ge n+1$. \\
Therefore, we have
\begin{eqnarray*}
&& \sum_{i=1}^n \sum_{j\neq i} \sum_{k, l, h}\frac{ \partial^3 s_{\gamma_{ij}} (\bs{\gamma}^*, \bar{\theta}^* )}{ \partial \theta_k \partial \theta_l \partial \theta_h } \\
&=&\frac{1}{2} \frac{1}{\sqrt{N}} \sum_{i=1}^n \sum_{j=1}^{n-1} Z_{ij}[p_{ij}(1-p_{ij})(1-6p_{ij}
+6p_{ij}^2)(\hat{\alpha}_i-\alpha_i^*)^2(\hat{\beta}_j-\beta_j^*)+ \\
&&p_{ji}(1-p_{ji})(1-6p_{ji}+6p_{ji}^2)(\hat{\alpha}_i-\alpha_i^*)(\hat{\beta}_j-\beta_j^*)^2].
\end{eqnarray*}
Let $\lambda_n = \|\bs{\hat{\theta}}^* - \bs{\theta}^* \|_\infty$. Note that $Z_{ij}$ lies in a
compact set $\Z$, and $p_{ij}(1-p_{ij})\le 1/4$, and $|(1-6p_{ij}+6p_{ij}^2)|\le 6$. By Lemma \ref{lemma:th4:d},
any element of $S_3$ is bounded above by
\begin{eqnarray*}
  \frac{ n(n-1)}{\sqrt{N}} \times   \frac{6}{4}\lambda_n^3   \times \sup_{z\in \Z} |z|
&=& 3\frac{ n(n-1) }{\sqrt{n(n-1)}} \times  \frac{ C^3 (\log n)^{3/2}  e^{24\|\theta^*\|_\infty } }{ n^{3/2} }
\times \sup_{z\in Z} |z| \\
& = & O\bigg(\frac{ (\log n)^{3/2}e^{24\|\theta^*\|_\infty } }{ \sqrt{n} } \bigg)=o(1).
\end{eqnarray*}

Similar to the calculation of deriving the asymptotic bias in Theorem 4 in \cite{Graham:2017}, we have
$S_2=B_*+o_p(1)$, where
\begin{equation}\label{definition:Bstar}
B_*=\lim_{n\to\infty} \frac{1}{2\sqrt{N}} \left[\sum_{i=1}^n \frac{  \sum_{j\neq i} p_{ij}(1-p_{ij})(1-2p_{ij})Z_{ij} }
{  \sum_{j\neq i} p_{ij}(1-p_{ij}) }+\sum_{j=1}^n \frac{  \sum_{i\neq j} p_{ij}(1-p_{ij})(1-2p_{ij})Z_{ij} }
{  \sum_{i\neq j} p_{ij}(1-p_{ij}) } \right].
\end{equation}

By Lemma \ref{lemma:th4:c}, similar to deriving the asymptotic expression of $S_1$ in \cite{Graham:2017}, we have
\[
S_1 = \frac{1}{\sqrt{N}} \sum_{i=1}^n \sum_{j\neq i} s_{\gamma_{ij}}^* (\bs{\gamma}^*, \bs{\theta}^*) + o_p(1),
\]
Therefore, it shows that equation \eqref{eq:gamma:asym:key} equal to
\begin{equation}\label{eq:proof:4-a}
\frac{1}{\sqrt{N}} \sum_{i=1}^n \sum_{j\neq i} s_{\gamma_{ij}}(\bs{\gamma}^*, \bs{\hat{\theta}}^* )
= \frac{1}{\sqrt{N}} \sum_{i=1}^n \sum_{j\neq i} s_{\gamma_{ij}}^* (\bs{\gamma}^*, \bs{\theta}^*) + B_* + o_p(1),
\end{equation}
with $\frac{1}{\sqrt{N}} \sum_{i=1}^n \sum_{j\neq i} s_{\gamma_{ij}}^* (\bs{\gamma}^*, \bs{\theta}^*)$ equivalent to the first two terms in \eqref{eq:gamma:asym:key}
and $B_*$ the probability limit of the third term in \eqref{eq:gamma:asym:key}.

Substituting \eqref{eq:proof:4-a} into \eqref{eq:theorem4:aa} then gives
\[
\sqrt{N}(\bs{\hat{\gamma}}- \bs{\gamma}^*) = I_*^{-1}(\gamma ) B_* + I_*^{-1}(\gamma) \frac{1}{\sqrt{N}} \sum_{i=1}^n \sum_{j\neq i}
s_{\gamma_{ij}}^* (\gamma^*, \theta^*) + o_p(1).
\]
Then Theorem \ref{theorem:covariate:asym} immediately follows from Lemma \ref{lemma:th4-b}.
\end{proof}


\begin{thebibliography}{}
\bibliographystyle{imsart-nameyear}

\setlength{\itemsep}{-1.5pt}


\bibitem[Adamic and Glance(2005)]{Adamic:Glance:2005}
{ Adamic, L. A.} and { Glance,  N.} (2005). The political blogosphere and the 2004 US Election.
{\it Proceedings of the WWW-2005 Workshop on the Weblogging Ecosystem}.

\bibitem[Amemiya(1985)]{Amemiya:1985}
{ Amemiya, T.} (1985). {\it Advanced Econometrics}. Cambridge, MA. Harvard University
Press.

\bibitem[Bader and Hogue(2003)]{Bader:Hogue:2003}
{ Bader, G. D.} and { Hogue, C. W. V.} (2003).
An automated method for finding molecular complexes in large protein interaction networks.
{\it BMC Bioinformatics}, 4:2, doi:10.1186/1471-2105-4-2.

\bibitem[Barab\'{a}si and Bonabau(2003)]{Barabasi:Bonabau:2003}
{ Barab\'{a}si, A. L.} and { Bonabau, E.} (2003). Scale-free networks. {\it Scientific
American}, 50--59.

\bibitem[Bradley and Terry(1952)]{bradleyterry52}
{Bradley, R. A.} and {Terry,
M. E.} (1952). The rank analysis of incomplete block designs I. The
method of paired comparisons. {\it Biometrika}, \textbf{39}, 324--345.


\bibitem[Brown(1986)]{Brown:1986}
Brown, L. D. (1986).
Fundamentals of statistical exponential families with applications in statistical decision theory.
{\it Lecture Notes-Monograph Series}. Hayward, California.


\bibitem[Burt et al.(2013)]{Burt:Kilduff:Tasselli:2013}
{ Burt, R. S., Kilduff, M.,} and { Tasselli, S.} (2013).
Social Network Analysis: Foundations and Frontiers on Advantage.
{\it Annual Review of Psychology}, {\bf 64}, 527--547.


\bibitem[Cai et al.(2018)]{Cai:2018}
Cai, W., Guan, G., Pan, R., Zhu, X., and Wang, H. (2018).
Network linear discriminant analysis. {\it Computational Statistics} \& {\it Data Analysis}, {\bf 117}, 32--44.

\bibitem[Chatterjee and Diaconis(2013)]{Chatterjee:Diaconis:2013}
{ Chatterjee, S.} and { Diaconis, P.} (2013). Estimating and understanding exponential random
graph models. {\it The Annals of Statistics}, {\bf 41}, 2428--2461.


\bibitem[Chatterjee et al.(2011)]{Chatterjee:Diaconis:Sly:2011}
{ Chatterjee, S., Diaconis, P.}, and { Sly, A.} (2011).
Random graphs with a given degree sequence. {\it Annals of Applied Probability}, \textbf{21}, 1400--1435.

\bibitem[Cruz-Gonzalez et al.(2017)]{Cruz-Gonzalez:Fern¡äandez-Val:Weidner:(2017)}
Cruz-Gonzalez, M., Fern\'{a}ndez-Val, I., and Weidner, M. (2017). Probitfe and logitfe: Bias corrections for probit
and logit models with two-way fixed effects. {\it The Stata Journal}, To appear

\bibitem[Diesner and Carley(2005)]{Diesner:Carley:2005}
{ Diesner, J.} and { Carley, K. M.} (2005).
Exploration of Communication Networks from the Enron Email Corpus.
{\it Proceedings of Workshop on Link Analysis, Counterterrorism and Security},
SIAM International Conference on Data Mining,  3--14.

\bibitem[Dzemski(2014)]{Dzemski:2014}
{ Dzemski, A.} (2014). An empirical model of dyadic link formation in a network
with unobserved heterogeneity. Preprint. Available at \url{http://pseweb.eu/ydepot/semin/texte1314/JMP\%20ANDREAS\%20DZEMKI.pdf}.


\bibitem[Dzemski(2017)]{Dzemski:2017}
{ Dzemski, A.} (2017). An empirical model of dyadic link formation in a network
with unobserved heterogeneity. Working Papers in Economics, No 698.


\bibitem[Erosheva et al.(2007)]{Erosheva:Fienberg:Joutard:2007}
{Erosheva, E. A.}, {Fienberg, S. E.}, and {Joutard, C.} (2007). Describing disability
through individual-level mixture models for multivariate binary data. {\it The Annals of Applied Statistics},
{\bf 1}, 502--537.

\bibitem[Fellows and Handcock(2012)]{Fellows:Handcock:2012}
{ Fellows, I.} and { Handcock, M. S.} (2012).
Exponential-family Random Network Models.
Available at \url{http://arxiv.org/abs/1208.0121}.

\bibitem[Fienberg(2012)]{Fienberg:2012}
{ Fienberg, S. E.} (2012). A brief history of statistical models for network analysis and open challenges.
{\it Journal of Computational and Graphical Statistics}, {\bf 21}, 825--839.

\bibitem[Fienberg and Wasserman(1981)]{Fienberg:Wasserman:1981a}
{ Fienberg, S. E.} and { Wasserman, S. S.} (1981). Categorical data analysis
of single sociometric relations. {\it Sociological Methodology}, {\bf 12}, 156--192.

\bibitem[Fienberg and Rinaldo(2007)]{Fienberg:Rinaldo:2007}
{Fienberg, S. E.} and {Rinaldo, A.} (2007). Three centuries of categorical data analysis:
Log-linear models and maximum likelihood estimation. {\it Journal of Statistical Planning and Inference}, {\bf 137},
3430--3445.

\bibitem[Fienberg and Rinaldo(2012)]{Fienberg:Rinaldo:2012}
{ Fienberg, S. E.} and  { Rinaldo, A.} (2012). Maximum likelihood estimation in log-linear models.
{\it The Annals of Statistics}, {\bf 40}, 996--1023.

\bibitem[Fern\'{a}ndez-Val and Weidner(2016)]{FVW2016}
{ Fern\'{a}ndez-V\'{a}l, I.} and { Weidner, M.} (2016). Individual and time effects in nonlinear
panel models with large $N$, $T$. {\it Journal of Econometrics}, {\bf 192}, 291--312.

\bibitem[Goldenberg et al.(2009)]{Goldenberg2009}
{ Goldenberg, A.}, {Zheng, A. X.}, {Feinberg, S. E.}, and { Airoldi, E. M.} (2009). A survey of statistical network models. {\it Foundations and Trends in Machine Learning}, {\bf 2}, 129--233.

\bibitem[Graham(2017)]{Graham:2017}
{ Graham, B. S.} (2017). An econometric model of link formation with degree heterogeneity. {\it Econometrica}, {\bf 85}, 1033--1063.


\bibitem[Graham(2016)]{Graham:2016}
Graham B. S. (2016). Homophily and transitivity in dynamic network formation.
NBER Working Paper, No. 22186. Available at \url{http://www.nber.org/papers/w22186}.





\bibitem[Haberman(1974)]{Haberman:1974}
{Haberman, S. J.} (1974). {\it The Analysis of Frequency Data}. Univ. Chicago Press, Chicago,
IL.

\bibitem[Haberman(1977)]{Haberman1977}
{ Haberman, S. J.} (1977). Maximum likelihood estimates in exponential response models.
{\it The Annals of Statistics}, {\bf 5}, 815--841.

\bibitem[Hahn and Newey(2004)]{Hahn:Newey:2004}
{ Hahn, J.} and { Newey W.} (2004). Jackknife and analytical bias reduction for
nonlinear panel data models.  {\it Econometrica}, {\bf 72}, 1295--1319.

\bibitem[Handcock(2003)]{Handcock:2003}
{ Handcock, M. S.} (2003). Assessing degeneracy in statistical models of social networks,
Working Paper 39, Techenical report, Center for Statistics and the Social Sciences,
University of Washington.

\bibitem[Hillar and Wibisono(2013)]{Hillar:Wibisono:2013}
{ Hillar, C.} and { Wibisono, A.} (2013). Maximum entropy distributions on graphs.
Available at \url{http://arxiv.org/abs/1301.3321}.

\bibitem[Hoeffding(1963)]{Hoeffding:1963}
{ Hoeffding, W.} (1963). Probability inequalities for sums of bounded random variables.
{\it Journal of the American Statistical Association}, {\bf 58}, 13--30.

\bibitem[Hoff(2009)]{Hoff:2009}
{ Hoff, P. D.} (2009). Multiplicative latent factor models for description and prediction of social networks. {\it Computational and Mathematical Organization Theory}, {\bf 15}, 261--272.

\bibitem[Holland and Leinhardt(1981)]{Holland:Leinhardt:1981}
{ Holland, P. W.} and { Leinhardt, S.} (1981). An exponential family of probability distributions for directed graphs (with discussion).
{\it Journal of the American Statistical Association}, {\bf 76}, 33--65.

\bibitem[Holland, Laskey and Leinhardt(1983)]{Holland:Laskey:Leinhardt:1983}
Holland, P. W., Laskey, K. B., and Leinhardt, S. (1983). Stochastic blockmodels: First steps. {\it Social Networks.} {\bf 5}, 109-137.


\bibitem[Jochmans(2017)]{Jochmans:2017}
Jochmans, K. (2017). Semiparametric analysis of network formation. {\it Journal of Business} \& {\it Economic Statistics}, To appear.


\bibitem[Juodis(2013)]{Juodis:2013}
Juodis, A. (2013). A note on bias-corrected estimation in dynamic panel data models. {\it Economics Letters}, {\bf 118}, 435--438.

\bibitem[Karwa and Slavkovi\'{c}(2016)]{Karwa:Slakovic:2016}
{ Karwa, V.} and { Slavkovi\'{c}, A.} (2016). Inference using noisy degrees-Differentially private beta model and synthetic graphs.
{\it The Annals of Statistics},  {\bf 44}, 87--112.

\bibitem[Kolaczyk(2009)]{Kolaczyk:2009}
Kolaczyk, E. D. (2009). {\it Statistical Analysis of Network Data: Methods and Models.} New York, Springer.

\bibitem[Krivitsky et al.(2011)]{Krivitsky:Handcock:Morris:2011}
Krivitsky, P. N., Handcock, M. S., and Morris, M. (2011). Adjusting for network
size and composition effects in exponential-family random graph models.
{\it Statistical Methodology}, 8, 319--339.


\bibitem[Lazega(2001)]{Lazega:2001}
{ Lazega, E.} (2001). {\it The Collegial Phenomenon: The Social Mechanisms of
Cooperation Among Peers in a Corporate Law Partnership.} Oxford University
Press, Oxford.


\bibitem[Lang(1993)]{Lang:1993}
{ Lang, S.} (1993). {\it Real and Functional Analysis.} Springer.

\bibitem[Lewisa et al.(2012)]{Lewisa:Gonzaleza:Kaufmanb:2012}
{ Lewisa, K., Gonzaleza, M.}, and { Kaufmanb, J.} (2012).
Social selection and peer influence in an online social network.
{\it Proceedings of the National Academy of Sciences of the United
States of America}, {\bf 109}, 68--72.

\bibitem[Lo\'{e}ve(1977)]{Loeve:1977}
{ Lo\'{e}ve, M.} (1977). {\it Probability theory I}. 4th ed. Springer, New York.

\bibitem[McPherson et al.(2001)]{McPherson:Lynn:Cook:2001}
{ McPherson, M., Lynn, S. L.}, and { Cook, J. M.} (2001). Birds of a feather:
homophily in social networks. {\it Annual Review of Sociology}, {\bf 27}, 415--444.

\bibitem[Nepusz et al.(2012)]{Nepusz:Yu:Paccanaro:2012}
{ Nepusz, T., Yu, H.}, and { Paccanaro, A.} (2012).
Detecting overlapping protein complexes in protein-protein interaction networks.
{\it Nature methods}, {\bf 18}, 471--472.

\bibitem[Newman(2002)]{Newman:2002}
{ Newman, M. E. J.} (2002). Spread of epidemic disease on
networks. {\it Physics Review E}, {\bf 66}, 016128.


\bibitem[Neyman and Scott(1984)]{Neyman:Scott:1948}
Neyman, J. and Scott, E. (1948). Consistent estimates based on partially consistent
observations. {\it Econometrica}, {\bf 16}, 1--32.


\bibitem[Olhede and Wolfe(2012)]{Olhede:Wolfe:2012}
{ Olhede, S. C.} and { Wolfe, P. J.} (2012). Degree-based network models. Available at
\url{http://arxiv.org/abs/1211.6537}.

\bibitem[Perry and Wolfe(2012)]{Perry:Wolfe:2012}
{Perry, P. O.} and {Wolfe, P. J.} (2012). Null models for network data. Available at
\url{http://arxiv.org/abs/1201.5871}.

\bibitem[Rasch(1960)]{Rasch1960}
{ Rasch, G.} (1960). {\it Probabilistic Models For Some Intelligence And Attainment Tests.}
Copenhagen: Paedagogiske Institut.

\bibitem[Rinaldo et al.(2011)]{Rinaldo:Petrovic:Fienberg:2011}
{Rinaldo, A.}, {Petrovi\'{c}, S.}, and {Fienberg, S.} (2011). Maximum likelihood estimation
in network models. Technical report. Available at \url{http://arxiv.org/abs/1105.6145}.

\bibitem[Rinaldo et al.(2013)]{Rinaldo2013}
{ Rinaldo, A., Petrovi\'{c}, S.}, and { Fienberg, S. E.} (2013). Maximum likelihood estimation in the $\beta$-model.
{\it The Annals of Statistics}, {\bf 41}, 1085--1110.


\bibitem[Robins et, al.(2007a)]{Robins.et.al.2007a}
{ Robins, G., Pattison, P., Kalish, Y.,} and { Lusher, D.} (2007a). An introduction to exponential random
graph ($p^*$) models for social networks. {\it Social Networks}, {\bf 29}, 173--191.

\bibitem[Robins et, al.(2007b)]{Robins.et.al.2007b}
{ Robins, G., Snijders, T., Wang, P., Handcock, M.}, and { Pattison, P.} (2007b). Recent
developments in exponential random graph (p*) models for social networks.
{\it Social Networks}, {\bf 29}, 192--215.

\bibitem[Sadeghi and Rinaldo(2014)]{Sadeghi:Rinaldo:2014}
{ Sadeghi, K.} and { Rinaldo, A.} (2014).
Statistical models for degree distributions of networks, NIPS 2014 Workshop ``From Graphs to Rich Data".
Available at \url{http://arxiv.org/abs/1411.3825}.

\bibitem[Schweinberger and Handcock(2015)]{Schweinberger:Handcock:2015}
{ Schweinberger, M.} and { Handcock, M. S.} (2015).
Local dependence in random graph models: characterization, properties and statistical inference.
{\it Journal of the Royal Statistical Society: Series B}, {\bf 77}, 647--676.

\bibitem[Shalizi and Rinaldo(2013)]{Shalizi:Rinaldo:2013}
{ Shalizi, C. R.} and { Rinaldo, A.} (2013). Consistency under sampling of exponential random graph models.
{\it The Annals of Statistics}, {\bf 41}, 508--535.


\bibitem[Simons and Yao(1999)]{simons1999}
{Simons, G.} and {Yao, Y. C.} (1999). Asymptotics when
the number of parameters tends to infinity in the Bradley-Terry
model for paired comparisons. {\it The Annals of Statistics}, \textbf{27}, 1041--1060.

\bibitem[Van Duijn et al.(2004)]{VanDuijnSnijdersZijlstra:2004}
{ Van Duijn, M. A. J., Snijders, T. A. B.}, and { Zijlstra, B. J. H.} (2004).
$p_2$: a random effects model with covariates for directed graphs.
{\it Statistica Neerlandica}, {\bf 58}, 234--254.

\bibitem[Wang and Wong(1987)]{Wang:Wong:1987}
Wang, Y. J. and Wong, G. Y. (1987). Stochastic blockmodels for directed graphs.
{\it Journal of the American Statistical Association}, {\bf 82}, 9--19.


\bibitem[Wu(1997)]{Wu:1997}
{ Wu, N.} (1997). {\it The Maximum Entropy Method.} New York, Springer.

\bibitem[Yan and Leng(2015)]{Yan:Leng:2013}
{ Yan, T.} and { Leng, C.} (2015). A simulation study of the $p_1$ model. {\it Statistics and Its Interface}, {\bf 8}, 255--266.

\bibitem[Yan et al.(2016)]{Yan:Leng:Zhu:2016}
{ Yan, T.}, { Leng, C.}, and { Zhu, J.} (2016). Asymptotics in directed exponential random graph models with an increasing bi-degree sequence. {\it The Annals of Statistics}, {\bf 44}, 31--57.

\bibitem[Yan and Xu(2013)]{Yan:Xu:2013}
{ Yan, T.} and { Xu, J.} (2013).
A central limit theorem in the $\beta$-model for undirected random graphs with a diverging number of vertices.
{\it Biometrika}, {\bf 100}, 519--524.

\bibitem[Yin(2015)]{Yin:2015}
{ Yin, M.} (2015).
A detailed investigation into near degenerate exponential random graphs.
Preprint. Available at \url{http://arxiv.org/abs/1512.06563}.

\end{thebibliography}
\end{document}